\documentclass[runningheads]{llncs}
\usepackage[T1]{fontenc}
\usepackage{graphicx}
\usepackage{amsmath} 
\usepackage{amsthm} 

\usepackage{mathtools} 

\usepackage{thmtools} 
\usepackage{thm-restate}

\usepackage{float}
\usepackage{subcaption}

\usepackage{dblfloatfix}

\usepackage{ifthen}

\usepackage[dvipsnames]{xcolor}

\usepackage{environ}

\usepackage{tikz}
\usetikzlibrary{positioning}
\usetikzlibrary{arrows.meta}
\usetikzlibrary{automata}
\usetikzlibrary{matrix}
\usetikzlibrary{backgrounds}

\usepackage{xcolor}

\usepackage{ulem}

\DeclareMathOperator{\width}{width}
\DeclareMathOperator{\argmin}{argmin}

\begin{document}
\title{Indexing Finite-State Automata Using Forward-Stable Partitions}
%
%
\author{Ruben Becker\orcidID{0000-0002-3495-3753} \and
Sung-Hwan Kim\orcidID{0000-0002-1117-5020} \and
Nicola Prezza\orcidID{0000-0003-3553-4953}  
\and Carlo Tosoni\orcidID{0009-0009-5149-2416} 
}

\authorrunning{Becker et al.}
%

\institute{DAIS, Ca' Foscari University of Venice, Italy}
%
\maketitle 
%

\newcommand{\version}{arXiv}


\ifthenelse{\equal{\version}{SPIRE}}{

\begin{abstract}
An index on a finite-state automaton is a data structure able to locate specific patterns on the automaton's paths and consequently on the regular language accepted by the automaton itself. Cotumaccio and Prezza [SODA '21], introduced a data structure able to solve pattern matching queries on automata, generalizing the famous FM-index for strings of Ferragina and Manzini [FOCS '00]. The efficiency of their index depends on the width of a particular partial order of the automaton’s states, the smaller the width of the partial order, the faster is the index. However, computing the partial order of minimal width is NP-hard. This problem was mitigated by Cotumaccio [DCC '22], who relaxed the conditions on the partial order, allowing it to be a partial preorder. This relaxation yields the existence of a unique partial preorder of minimal width that can be computed in polynomial time. In the paper at hand, we present a new class of partial preorders and show that they have the following useful properties: (i) they can be computed in polynomial time, (ii) their width is never larger than the width of Cotumaccio's preorders, and (iii) there exist infinite classes of automata on which the width of Cotumaccio's pre-order is linearly larger than the width of our preorder.




\end{abstract}

\keywords{Nondeterministic Finite Automata \and Graph Indexing \and Forward-Stable Partitions \and FM-index.}


\section{Introduction}


The Burrows-Wheeler transform (BWT) is a famous reversible string transformation~\cite{burrows1994block}. While the BWT was initially conceived as a compression tool, it has subsequently been used to implement the FM index~\cite{manziniFM} that is able to locate patterns in almost optimal time, while efficiently compressing the strings at the same time. This indexing strategy has been extended to some finite automata (or directed edge-labeled graphs), the so-called \textit{Wheeler graphs}, by Gagie et al.~\cite{gagie2017wheeler}. Wheeler graphs are a particular class of automata that can be succinctly encoded by a representation that allows pattern matching queries on the strings labeling directed paths in the automaton in almost optimal time. This indexing property relies on the fact that the states of the automaton can be totally ordered in a way that is consistent with the co-lexicographic order of the strings accepted by the automaton's states. Such an order of the states is called a \textit{Wheeler order}. As not all automata admit a Wheeler order, this strategy may fail to apply and hence not all finite automata are Wheeler. Cotumaccio and Prezza~\cite{cotumaccio2021indexing} extended the notion to arbitrary finite automata by not restricting the state order to be total, instead allowing arbitrary partial orders, called \textit{co-lex orders}. As any finite automaton admits such a co-lex order, this enables us to build an efficient index for the language accepted by any automaton in a very similar fashion as one does with the original FM index for a given string. The co-lex order however is not in general unique and, furthermore, the choice of the co-lex order is of crucial importance for the efficiency of the index. More precisely, the efficiency is directly characterized by the \textit{width} of the particular co-lex order: the lower the width of the co-lex order is, the faster pattern matching queries are and the smaller the index is. Unfortunately, computing the co-lex order of minimal width is NP-hard~\cite{cotumaccio2021indexing}. Cotumaccio~\cite{cotumaccio2022graphs} introduced \textit{co-lex relations} by relaxing the requirements on co-lex orders, allowing them to be reflexive relations, i.e., unrequiring antisymmetry and transitivity. This relaxation yields that a unique maximum co-lex relation always exists. Moreover, this maximum is a preorder (i.e., transitivity is recovered through maximality). Maybe most importantly though, the maximum co-lex relation can be computed in polynomial time (more precisely in quadratic time in the number of transitions). The width of the maximum co-lex relation is at most the width of every co-lex order and may be asymptotically smaller. Hence, the index built based on the maximum co-lex relation is never slower and can be asymptotically 
faster than those built based on co-lex orders~\cite{cotumaccio2021indexing}.

\paragraph{Our contribution.}
In this paper, we present a new category of preorders that we call \textit{coarsest forward-stable co-lex (CFS) orders}. CFS orders are equally useful as maximum co-lex relations for constructing indices on automata. Similar to co-lex orders being the generalization of Wheeler orders to width larger than one, CFS orders are the generalization of Wheeler preorders (introduced by Becker et al.~\cite{becker2023sorting}) to width larger than one. As shown by Cotumaccio for maximum co-lex relations, we can show that our CFS orders always exist and are unique. Furthermore, for a given automaton, the unique CFS order can be computed in quadratic time in the number of transitions as well. Most importantly however, we prove that the width of the 
CFS order is never larger than the width of the maximum co-lex relation and, 
in some cases, asymptotically smaller. 

In Figure~\ref{fg:ex:3}, we show the relationships among the previously  mentioned concepts for building indices on automata. All previously known orders shown in Figure~\ref{fg:ex:3} are formally introduced in Section~\ref{sec:not}, where we describe the preliminaries for this work. In Section~\ref{sec:CSF} we introduce coarsest forward-stable co-lex (CFS) orders and in Section~\ref{sec:CSF CR} we prove the claimed properties of these preorders. Some proofs in this article are deferred to an appendix due to space limitations.

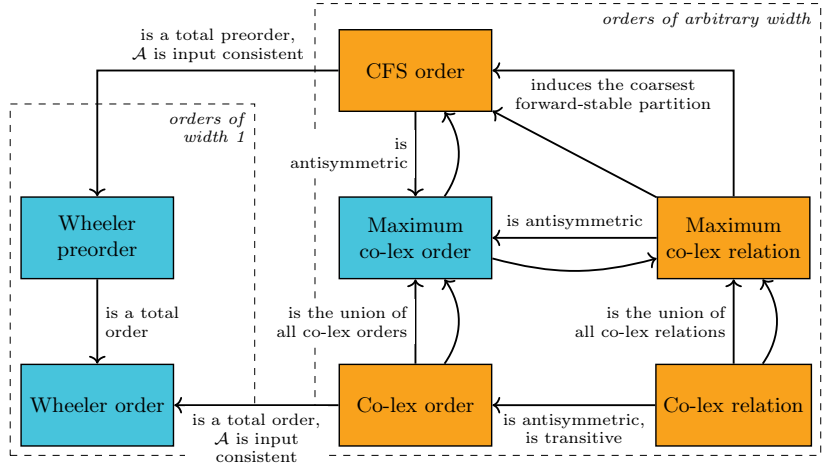
\begin{figure}[ht]
\begin{center}
\resizebox{0.9\textwidth}{!}{
\begin{tikzpicture}
[rel/.style={shape = rectangle, align = center, minimum height = 3.75em, minimum width = 7em, font=\footnotesize, draw=black, semithick},
col1/.style={fill=SkyBlue},
col2/.style={fill=YellowOrange},
edges/.style={font=\scriptsize, thick},
widthcol/.style={draw=BlueViolet},
cat/.style={dashed, thin},
catnod/.style={font=\scriptsize}]

    \node[rel, col2] (1) at (9.65,0) {Co-lex relation};
    \node[rel, col2] (2) at (4.9,0) {Co-lex order};
    \node[rel, col2] (3) at (9.65,2.5) {Maximum\\co-lex relation};
    \node[rel, col1] (4) at (4.9,2.5) {Maximum\\co-lex order};
    \node[rel, col2] (5) at (4.9,5) {CFS order};
    \node[rel, col1] (6) at (0.15,2.5) {Wheeler\\preorder};
    \node[rel, col1] (7) at (0.15,0) {Wheeler order};

    \begin{scope}[on background layer]
        \draw[cat] (-1.15, 4.5) rectangle (2.5,-0.75);
        \node[catnod, align = right] at (1.8, 4.15) {\textit{orders of}\\\textit{width 1}};
        \draw[cat] (3.4,-0.75) rectangle (10.95, 6);
        \node[catnod] at (9.2,5.75) {\textit{
        orders of arbitrary width
        }};
    \end{scope}
    
    \draw[->, edges] (1) edge node [below, align = center] {
    is antisymmetric,\\is transitive
    } (2);
    \draw[->, edges] (1) edge node [left, align = right] {
    is the union of\\all co-lex relations
    } (3);
    \draw[->, edges] (2) edge node [left, align = right, fill = white] {
    is the union of\\all co-lex orders
    } (4);
    \draw[->, edges] (3) edge node [above] {
    is antisymmetric
    } (4);
    \path[->, edges] (5) edge node [left, align = right, fill = white] {
    is\\antisymmetric
    } (4);
    
    \draw[->, edges] (3) --  (9.65,5) -- node [below, align = center] {
    induces the coarsest\\forward-stable partition
    } (5);

    \draw[->, edges] (5) -- node [above, align = center]  {
    is a total preorder,\\$\mathcal{A}$ is input consistent
    } (0.15,5) -- (6);
    
    \draw[->, edges] (6) edge node [right, align = left] {
    is a total\\order
    } (7);

    \draw[->, edges] (2) edge node [below, align = center, fill=white]{
    is a total order,\\$\mathcal{A}$ is input\\consistent
    } (7);

    \draw[->, edges, widthcol, bend right = 35] (1) edge (3);

    \draw[->, edges, widthcol, bend right = 35] (2) edge (4);

    \draw[->, edges, widthcol, bend right = 35] (4) edge (5);

    \draw[->, edges, widthcol, bend right = 15] (4) edge (3);

    \draw[->, edges, widthcol] (3) edge (5);
    
\end{tikzpicture}
}
\end{center}
\caption{An NFA $\mathcal{A}$ is \textit{input consistent} if for each state $u$ in $\mathcal{A}$, all incoming edges of $u$ are labeled with the same character. 
We show the connections between the different relations described. 
A relation is orange, if every automaton always admits an instance of that relation, and it is blue otherwise. Orders on the left, i.e., \textit{Wheeler preorders} and \textit{Wheeler orders}, are of width 1, while the others may be of arbitrary width. 
A blue edge $A \rightarrow B$ means that any relation of type $A$ has always a width larger than or equal to a relation of type $B$. A black edge $A \xrightarrow{c} B$ means a relation of type $A$ is also a relation of type $B$ if it satisfies the requirements $c$. In this case, a relation of type $B$ is always also a relation of type $A$ with the following exceptions: (i) The \textit{coarsest forward-stable co-lex order} may not be equal to the \textit{maximum co-lex relation}. (ii) If the \textit{maximum co-lex order} exists, then it may not be equal to the \textit{coarsest forward-stable co-lex order}. (iii) A \textit{Wheeler order} may not be equal to the \textit{Wheeler preorder}. All implications either directly follow from their definitions or are proved in Appendix~\ref{appendix: relationships}.}
\label{fg:ex:3}
\end{figure}

\section{Preliminaries}\label{sec:not}

\paragraph{Nondeterministic finite automata.}
A nondeterministic finite automaton (NFA) is a 4-tuple $(Q,\delta,\Sigma,s)$, where $Q$ represents the set of the states, $\delta : Q \times \Sigma \rightarrow 2^{Q}$ is the automaton's transition function, $\Sigma$ is the alphabet and $s \in Q$ is the initial state. The standard definition of NFAs includes also a set of final states that we omit since we are not concerned in distinguishing between final and non-final states. We assume the alphabet $\Sigma$ to be effective, i.e., each character of $\Sigma$ labels at least one edge of the transition function. A deterministic finite automaton (DFA) is an NFA such that each state has at most one outgoing edge labeled with a given character. Given an NFA $\mathcal{A} = (Q, \delta, \Sigma, s)$, for a state $u \in Q$, and a character $a \in \Sigma$, we use the shortcut $\delta_{a}(u)$ for $\delta(u,a)$. For a set $S \subseteq Q$, we define $\delta_{a}(S) \coloneqq \bigcup_{u \in S}\delta_{a}(u)$. The set of finite strings over $\Sigma$, denoted by $\Sigma^{*}$, is the set of finite sequences of letters from $\Sigma$. We extend the transition function $\delta$ to the elements of $\Sigma^{*}$ in the following way: for $\alpha \in \Sigma^{*}$ and $u \in Q$ we define $\delta(u,\alpha)$ recursively as follows. If $\alpha = \varepsilon$ (i.e. $\alpha$ is the empty string) then $\delta(u,\alpha)  = \{u\}$. Otherwise, if $\alpha = \alpha'a$, with $\alpha' \in \Sigma^{*}$, $a \in \Sigma$, then $\delta(u,\alpha) = \bigcup_{v\in \delta(u,\alpha')}\delta(v,a)$. 


\begin{definition}[Strings reaching a state]\label{2:df:reac_str}
    Given an NFA $\mathcal{A}=(Q,\delta,\Sigma,s)$, the set of strings reaching a state $u \in Q$ is defined as $I_{u} = \{ \alpha \in \Sigma^{*} : u \in \delta(s,\alpha)\}$.
\end{definition}

We say that $I_{u}$ is the regular language recognized by state $u$, while the regular language recognized by $\mathcal{A}$ is the union of all the regular languages recognized by all $\mathcal{A}$'s states. Regarding the NFAs we treat, we make the following assumptions: (i) We assume that every state is reachable from the initial state. (ii) We assume that the initial state has no incoming edges. (iii) We do not require each state to have an outgoing edge for all possible labels. None of these assumptions are restrictive, since any NFA can be modified to satisfy these assumptions without changing its accepted language. Given an NFA $\mathcal{A} = $($Q$, $\delta$, $\Sigma$, $s$), and a state $u \in Q$, $u \neq s$, we denote with $\lambda(u)$ the set of the characters of $\Sigma$ that label the incoming edges of $u$. If $u = s$, we define $\lambda(u) = \{\#\}$, where $\# \notin \Sigma$.

\paragraph{Forward-stable partitions.}

Given a set $U$, a partition $\mathcal{U} = \{U_{i}\}_{i=1}^{k}$ of $U$ is a set of pairwise disjoint non-empty sets $\{U_{1}, ..., U_{k}\}$ whose union is $U$. We call the sets $U_{1}, ... ,U_{k}$ the parts of $\mathcal{U}$. Given two partitions $\mathcal{U}$ and $\mathcal{U}'$, we say that $\mathcal{U}'$ is a refinement of $\mathcal{U}$ if every part of $\mathcal{U}'$ is contained in a part of $\mathcal{U}$. Note that every partition is a refinement of itself. We use the concept of forward-stable partitions, see also the work of Alanko et al. \cite[Section 4.2]{alanko2021wheeler}.

\begin{definition}[Forward-Stability]
Given an NFA $\mathcal{A} = (Q,\delta,\Sigma,s)$ and two sets of states $S$, $T \subseteq Q$, we say that $S$ is \textit{forward-stable} with respect to $T$, if, for all $a \in \Sigma$, $S \subseteq \delta_{a}(T)$ or $S \cap \delta_{a}(T) = \emptyset$ holds. A partition $\mathcal{Q}$ of $\mathcal{A}$'s states is \textit{forward-stable} for $\mathcal{A}$, if, for any two parts $S$, $T \in \mathcal{Q}$, it holds that $S$ is forward-stable with respect to $T$.
\end{definition}

Given an NFA $\mathcal{A} = (Q,\delta,\Sigma,s)$ we say that $\mathcal{Q}$ is the \textit{coarsest} forward-stable partition for $\mathcal{A}$, if for every forward-stable partition $\mathcal{Q}'$ for $\mathcal{A}$, it holds that $\mathcal{Q}'$ is a refinement of $\mathcal{Q}$. It is easy to demonstrate that for each automaton $\mathcal{A}$ there exists a unique coarsest forward-stable partition (a proof can be found in the work of Becker et al. \cite[Appendix A]{becker2023sorting}). An example of the coarsest forward-stable partition for an NFA can be found in Figure \ref{fg:ex:2}. An interesting property about forward-stable partitions is the following: let $\mathcal{Q}$ be a forward-stable partition for an NFA $\mathcal{A}=(Q,\delta,\Sigma,s)$. If $u,v \in Q$ belong to the same part of $\mathcal{Q}$, then $u$ and $v$ are reached by the same set of strings; a proof of this property can be found in \cite[Lemma 4.7]{alanko2021wheeler}. However, the reverse does not necessarily hold: we may have states in different parts of a forward-stable partition $\mathcal{Q}$ that are reached by the same set of strings (even if $\mathcal{Q}$ is the coarsest forward-stable partition of $\mathcal{A}$). An example of this fact can be found in Figure 1 of the article by Becker et al. \cite{becker2023sorting}.  



\paragraph{Relations.}

Given a set $U$, a \textit{relation} $R \subseteq U \times U$ over $U$ is a set of ordered pairs of elements from $U$. For two elements $u$, $v$ from $U$, we write $uRv$ to denote that $(u,v) \in R$. A \textit{partial order} over a set $U$ is a relation that satisfies reflexivity, antisymmetry, and transitivity. If a partial order satisfies also connectedness (a relation over $U$ satisfies connectedness, if for each distinct $u, v \in U$, either $(u,v) \in R$ or $(v,u) \in R$ holds) then it is a \textit{total order}. A \textit{partial  preorder} over a set $U$ is a relation that satisfies reflexivity and transitivity, but not necessarily antisymmetry. A \textit{total preorder} is a partial preorder over a set $U$ that satisfies also connectedness. An \textit{equivalence relation} over a set $U$ is a relation that satisfies reflexivity, symmetry and transitivity. In this paper we use the symbol $\leq$ to denote both partial orders and partial preorders and the symbol $\sim$ for equivalence relations. Given an equivalence relation $\sim$ over a set $U$, we denote with $[u]_{\sim}$ the equivalence class of the element $u \in U$ with respect to $\sim$, i.e., $[u]_{\sim} \coloneqq \{v \in U : u \sim v\}$. We denote with $U/_{\sim}$ the partition of $U$ consisting of all equivalence classes $[u]_{\sim}$, for $u \in U$. We may not specify the set $U$ over which the relation is defined, if it is clear from the context. Given a partial order or a partial preorder $\leq$ over a set $U$, a set $L \subseteq U$ is an \textit{antichain} according to $\leq$ if for every $u, v \in L$, with $u \neq v$, pairs $(u,v)$ and $(v,u)$ do not belong to $\leq$. The \textit{width} of $\leq$, denoted by $\width(\leq)$, is the size of the largest antichain for $\leq$. Note that the width is equal to 1 if and only if $\leq$ is also a total order. A partial preorder $\leq$ over a set $U$ induces an equivalence relation $\sim$ over $U$ in the following way; For $u, v \in U$, we define $u \sim v$ if and only if $u \leq v$ and $v \leq u$. Moreover, a partial preorder $\leq$ over a set $U$ induces a partial order $\leq'$ over the set $U/_{\sim}$ (where $\sim$ is the equivalence induced by $\leq$) defined as $[u]_{\sim} \leq' [v]_{\sim}$ if and only if $u \leq v$. In addition, if $\leq$ is also a total preorder, then $\leq'$ is also a total order, since it is easy to observe that if $\leq$ satisfies connectedness, then also $\leq'$ satisfies connectedness. Given a partial order or a partial preorder $\leq$ over a set $U$, we define the symbol $<$ in the following way: for each $u, v \in U$, $u < v$ holds if and only if $u \leq v$ and $\neg(v \leq u)$. Finally, given an NFA $\mathcal{A} = (Q,\delta,\Sigma,s)$ and an equivalence relation $\sim$ over $Q$, we define the quotient automaton $\mathcal{A}/_{\sim}$ as follows:

\begin{definition}[Quotient automaton]\label{2:def:qa}
    Let $\mathcal{A} = (Q,\delta,\Sigma,s)$ be an NFA and $\sim$ an equivalence relation over $Q$. The \textit{quotient automaton} $\mathcal{A}/_{\sim}=(Q/_{\sim},\delta/_{\sim},\Sigma,s/_{\sim})$ of $\mathcal{A}$ is the NFA defined 
    by letting 
    $Q/_{\sim} := \{[u]_{\sim} : u \in Q\}$, 
    $\delta/_{\sim}([u]_{\sim},a) := \{[v]_{\sim} : (\exists v' \in [v]_{\sim}) (\exists u' \in [u]_{\sim}) (v' \in \delta(u', a))\}$, and 
    $s/_{\sim} := [s]_{\sim}$.
\end{definition}

An example of a quotient automaton of an NFA can be found in Figure \ref{fg:ex:2}.

\begin{figure}[ht]
\begin{center}
\resizebox{0.9\textwidth}{!}{
\begin{tikzpicture}
[dim/.style={minimum size=2.6em}, dots/.style={text centered}, 
scale=1, dim2/.style={minimum size=2.8, scale=0.85}, dots/.style={text centered}]
    \node[initial,state,initial where=below,dim] (0) at (0,-4) {$u_{0}$};
    \node[state,dim] (1) at (-2.3,-3.5) {$u_{1}$};
    \node[state,dim] (2) at (-1,-2.6) {$u_{2}$};
    \node[state,dim] (3) at (1,-2.6) {$u_{3}$};
    \node[state,dim] (4) at (2.3,-3.5) {$u_{4}$};
    \node[state,dim] (5) at (-2.2,-1.5) {$u_{5}$};
    \node[state,dim] (6) at (2.2,-1.5) {$u_{6}$};


    \path [->] (0) edge [bend left = 30] node [below] {$a$} (1);
    \path [->] (0) edge [bend left = 30] node [left] {$a$} (2);
    \path [->] (1) edge [loop left] node [left] {$a$} (1);
    \path [->] (2) edge [loop left] node [left] {$a$} (2);
    \path [->] (0) edge [bend right = 30] node [right] {$a$} (3);
    \path [->] (0) edge [bend right = 30] node [below] {$a$} (4);
    \path [->] (2) edge [bend left = 20] node [above] {$b$} (5);
    \path [->] (2) edge [bend left = 10] node [above] {$b$} (6);
    \path [->] (4) edge [bend right = 10] node [right] {$b$} (6);
    \path [->] (3) edge [bend right = 10] node [above] {$b$} (5);
    \path [->] (5) edge [bend left = 38] node [above] {$b$} (6);
    \path [->] (6) edge [bend right = 19] node [below] {$b$} (5);

    \begin{scope}[xshift=6.3cm]
    \node[initial,state,initial where=below,dim2] (0) at (0,-4) {$[u_{0}]_{\sim}$};
    \node[state,dim2] (1) at (-1.7,-2.5) {$[u_{1}]_{\sim}$};
    \node[state,dim2] (2) at (1.7,-2.5) {$[u_{3}]_{\sim}$};
    \node[state,dim2] (3) at (0,-1.0) {$[u_{5}]_{\sim}$};


    \path [->] (0) edge [bend left = 20] node [below] {$a$} (1);
    \path [->] (0) edge [bend right = 20] node [below] {$a$} (2);
    \path [->] (1) edge [loop left] node [left] {$a$} (1);
    \path [->] (1) edge [bend left = 20] node [above] {$b$} (3);
    \path [->] (2) edge [bend right = 20] node [above] {$b$} (3);
    \path [->] (3) edge [loop below] node [below] {$b$} (3);
    \end{scope}
\end{tikzpicture}
}
\end{center}
\caption{An NFA $\mathcal{A}=(Q,\delta,\Sigma,s)$ on the left. On the right, the corresponding quotient automaton $\mathcal{A}/_{\sim_{FS}}=(Q/_{\sim_{FS}},\delta/_{\sim_{FS}},\Sigma,s_{\sim_{FS}})$ of $\mathcal{A}$ for the coarsest forward-stable partition $Q/_{\sim_{FS}} = \{ \{u_{0}\}, \{u_{1}, u_{2}\}, \{u_{3}, u_{4}\}, \{u_{5}, u_{6}\} \}$.}
\label{fg:ex:2}



%
%

\label{4:fg:ex2}
\end{figure}
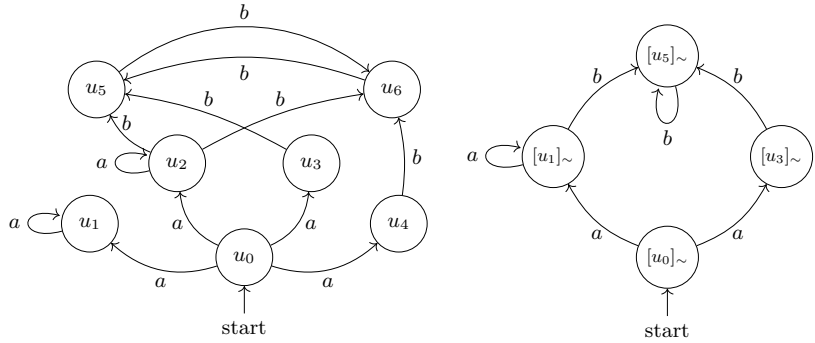

\paragraph{Wheeler and Quasi-Wheeler NFAs.}

The notion of Wheeler NFAs was first introduced by Gagie et al. \cite[Definition 1]{gagie2017wheeler}. Wheeler NFAs are a particular class of NFAs that can be endowed with a particular total order of their states that allows them to be efficiently compressed and indexed. From now on, we assume that given an NFA $\mathcal{A}=(Q,\delta,\Sigma,s)$ there exists a total order $\leq$ among the elements of the set $\Sigma \cup \{\#\}$, where $\# < a$ for each $a \in \Sigma$. In this paper we deliberately overload the notation regarding the symbol $\leq$, i.e., we will use the same symbol for orders on integers, the alphabet $\Sigma \cup \{ \# \}$ and states $Q$ of automata. It will still always be clear from the context which order we refer to.

\begin{definition}[Wheeler NFAs]
Let $\mathcal{A} = (Q,\delta,\Sigma,s)$ be an NFA. A Wheeler order $\leq$ of $\mathcal{A}$ is a total order over $Q$ such that state $s$ precedes all states in $Q \setminus \{s\}$, and, for any pair $u \in \delta_{a}(u')$ and $v \in \delta_{a'}(v')$:
\begin{enumerate}
    \item If $a < a'$, then $u < v$.
    \item If $a = a'$ and $u' < v'$, then $u \leq v$. 
\end{enumerate}
We say that $\mathcal{A}$ is a \textit{Wheeler NFA}, if there exists a \textit{Wheeler order} $\leq$ of $\mathcal{A}$

\end{definition}

In other words, A Wheeler order is a total order over $Q$ that sorts the states of $Q$ according to the set strings that reach them (see Definition \ref{2:df:reac_str}). The possibly largest drawback of Wheeler NFAs is that, recognizing whether a given NFA is Wheeler or not is an NP-complete problem~\cite{gibney2019hardness}. For this reason, Becker et al.~\cite[Definition 8]{becker2023sorting} proposed a relaxed version of the problem, introducing the notion of quasi-Wheeler NFAs. At this point, given an NFA $\mathcal{A}=(Q,\delta,\Sigma,s)$, we define as $\sim_{FS}$ the equivalence relation over $Q$, such that for each $u,v \in Q$, $u \sim_{FS} v$ holds if and only if $u$ and $v$ belong to the same part of the coarsest forward-stable partition for $\mathcal{A}$. Consequently, partition $Q/_{\sim_{FS}}$ is the coarsest forward-stable partition for $\mathcal{A}$

\begin{definition}[Quasi-Wheeler NFAs]\label{2:def:qw}
Let $\mathcal{A} = (Q,\delta,\Sigma,s)$ be an NFA. A Wheeler preorder $\leq$ of $\mathcal{A}$ is a total preorder over $Q$ such that:
\begin{enumerate}
    \item $\sim_{FS}$ is the equivalence relation induced by $\leq$.
    \item The quotient automaton $\mathcal{A}/_{\sim_{FS}}$ is a Wheeler NFA, and the total order $\leq'$ induced by $\leq$ is a Wheeler order for $\mathcal{A}/_{\sim_{FS}}$.
\end{enumerate}
We say that $\mathcal{A}$ is a quasi-Wheeler NFA, if there is a Wheeler preorder $\leq$ of $\mathcal{A}$.
\end{definition}

Due to the fact that states in the same part of a forward-stable partition are reached by the same set of strings, it is easy to observe that the automata $\mathcal{A}$ and $\mathcal{A}/{\sim_{FS}}$ of Definition \ref{2:def:qw} recognize the same language. Thus, from an indexing perspective the NFA $\mathcal{A}/_{\sim_{FS}}$ is equally useful as $\mathcal{A}$: an index for $\mathcal{A}/_{\sim_{FS}}$ is also an index for $\mathcal{A}$. Moreover, every Wheeler NFA is also a quasi-Wheeler NFA, and there exist NFAs that are quasi-Wheeler but not Wheeler. This implies that the class of quasi-Wheeler NFAs is strictly larger than the class of Wheeler NFAs (for a formal proof see Lemma 9 and Figure 1 of the article of Becker et al.~\cite{becker2023sorting}). Probably, the most important difference between Wheeler NFAs and quasi-Wheeler NFAs is that the latter can be recognized in polynomial time. In fact, Becker et al. \cite{becker2023sorting} proposed an algorithm that, given in input an NFA $\mathcal{A}$, is able to recognize if $\mathcal{A}$ is quasi-Wheeler in $O(\lvert \delta \rvert \log \lvert Q \rvert)$ time and, if this is the case, compute a Wheeler preorder for $\mathcal A$ in the same running time. This is achieved by extending the partition refinement framework of Paige and Tarjan~\cite{paige1987three} to compute the coarsest forward-stable partition for $\mathcal{A}$ and, at the same time, a total order $\leq$ over the parts of that partition (and thus over $\mathcal{A}$'s states), in such a way that if $\mathcal{A}$ is quasi-Wheeler, then $\leq$ is a Wheeler preorder of $\mathcal{A}$.

\paragraph{$p$-Sortable NFAs.}
Despite the fact that Wheeler NFAs can be indexed and compressed almost optimally, a major limitation regarding Wheeler NFAs is that the class of Wheeler NFAs is very limited. In fact, Wheeler languages, i.e., languages that are recognized by a Wheeler NFA, are star-free and closed only under intersection~\cite{alanko2021wheeler}. Moreover, also the amount of ``nondeterminism'' within a Wheeler NFA is bounded, since every Wheeler NFA admits an equivalent Wheeler DFA  of linear size \cite{alanko2020regular} (while in the general case there is an exponential blow-up of the number of states). For these reasons, Cotumaccio and Prezza \cite[Definition 3.1]{cotumaccio2021indexing} have extended the notion of Wheeler orders by introducing the concept of co-lex orders. From here on, given an NFA $A = (Q,\delta,\Sigma,s)$, and two states $u,v \in Q$, we say that $\lambda(u) \leq \lambda(v)$ holds if and only if for each $a \in \lambda(u)$ and for each $a' \in \lambda(v)$, it holds that $a \leq a'$.

\begin{definition}[Co-lex orders]\label{2:def:col_ord}
    Let $\mathcal{A} = (Q,\delta,\Sigma,s)$ be an NFA. A co-lex order of $\mathcal{A}$ is a partial order $\leq$ over $Q$ that satisfies the following two axioms:
    \begin{enumerate}
        \item For every $u, v \in Q$, if $u < v$, then $\lambda(u) \leq \lambda(v)$.
        \item For every pair $u \in \delta_{a}(u')$ and $v \in \delta_{a}(v')$, if $u < v$, then $u' \leq v'$.
    \end{enumerate}
\end{definition}

A main difference between Wheeler orders and co-lex orders is that every NFA $\mathcal{A}=(Q,\delta,\Sigma,s)$ admits a co-lex order; in fact, the relation $\leq\,\coloneqq \{(v,v) : v \in Q\}$ is a co-lex order of $\mathcal{A}$. Given an NFA $\mathcal{A}$ and a co-lex order $\leq$ of $\mathcal{A}$, we say that $\leq$ is the \textit{maximum} co-lex order of $\mathcal{A}$ if $\leq$ is equal to the union of all co-lex orders of $\mathcal{A}$. In general, an NFA does not always admit a maximum co-lex order. Cotumaccio and Prezza also define the notion of co-lexicographic width of an NFA \cite[Definition 3.3]{cotumaccio2021indexing}:

\begin{definition}[Co-lexicographic width]\label{3:df:co_lex_wdt}
    Let $\mathcal{A}$ be an NFA.
    \begin{enumerate}
        \item The NFA $\mathcal{A}$ is \textit{p-sortable} if there exists a co-lex order $\leq$ of $\mathcal{A}$ of width $p$. 
        \item The \textit{co-lexicographic width} $\bar{p}$ of $\mathcal{A}$ is the smallest $p$ for which $\mathcal{A}$ is p-sortable.
    \end{enumerate}
\end{definition}

From Definition \ref{3:df:co_lex_wdt}, it is possible to observe that if an NFA $\mathcal{A}$ is a Wheeler NFA, then the co-lexicographic width of $\mathcal{A}$ is equal to 1, since a Wheeler order is a particular type of co-lex order that is also a total order. As Cotumaccio and Prezza shows, the co-lexicographic width $\bar{p}$ of an NFA $\mathcal{A}$ measures the space and time complexity with which $\mathcal{A}$ can be encoded and indexed. However, the problem of determining the co-lexicographic width of a given NFA is known to be NP-hard. This follows from NP-hardness of Wheelerness.

\paragraph{Indexable Partial Preorders.}

At this point, we have seen that some particular NFAs $\mathcal{A}=(Q,\delta,\Sigma,s)$ (i.e. Wheeler NFAs) can be endowed with a particular total order over $Q$ (i.e. a Wheeler order) in order to build efficient indexes on top of them. Despite the fact that the problem of determining whether or not a general NFA $\mathcal{A}$ admits a Wheeler order is NP-complete, we have seen that it is possible to compute in polynomial time a particular total preorder $\leq$ over $Q$ such that the quotient automaton $\mathcal{A}/_{\sim_{FS}}=(Q/_{\sim_{FS}},\delta/_{\sim_{FS}},\Sigma,s/_{\sim_{FS}})$ (where $\sim_{FS}$ is the equivalence relation induced by $\leq$) has the following properties: $(i)$ $\mathcal{A}$ and $\mathcal{A}/_{\sim_{FS}}$ accept the same language, $(ii)$ If $\mathcal{A}$ is Wheeler, then also $\mathcal{A}/_{\sim_{FS}}$ is Wheeler, $(iii)$ if $\mathcal{A}/_{\sim_{FS}}$ is Wheeler, then the total order $\leq'$ over $Q/_{\sim_{FS}}$ induced by $\leq$ is a Wheeler order for $\mathcal{A}/_{\sim_{FS}}$. Now it is natural to wonder whether or not these reasonings can be generalized also to arbitrary NFAs; to this end, we introduce the notion of indexable partial preorders.

\begin{definition}[Indexable partial preorders]\label{2:def:indpp}
Let $\mathcal{A} = (Q,\delta,\Sigma,s)$ be an NFA of co-lexicographic width $\bar{p}$, and let $\leq$ be a partial preorder over $Q$. Consider $\mathcal{A}/_{\sim}=(Q/_{\sim},\delta/_{\sim},\Sigma,s/_{\sim})$ the quotient automaton defined by the equivalence relation $\sim$ induced by $\leq$. We say that $\leq$ is an \textit{indexable partial preorder} for $\mathcal{A}$ if the following requirements are satisfied:
\begin{enumerate}
    \item $\mathcal{A}$ and $\mathcal{A}/_{\sim}$ accept the same language.

    \item $\mathcal{A}/_{\sim}$ has co-lexicographic width $\bar{q}$, with $\bar{q} \leq \bar{p}$.

    \item The partial order $\leq'$ over $Q/_{\sim}$ induced by $\leq$ is a co-lex order of $A/_{\sim}$ of width $\bar{q}$.
\end{enumerate}
\end{definition}

It is clear that a Wheeler preorder of an automaton $\mathcal{A}$ is also an indexable partial preorder for $\mathcal{A}$. However, in the literature there exists another example of indexable partial preorders, namely the maximum co-lex relations of Cotumaccio~\cite{cotumaccio2022graphs}. Cotumaccio defined the concept of \textit{co-lex relations} of NFAs~\cite[Definition 1]{cotumaccio2022graphs}. Each NFA $\mathcal{A}$ admits a \textit{maximum} co-lex relation, that is the co-lex relation of $\mathcal{A}$ that is equal to the union of all co-lex relations of $\mathcal{A}$. Finally, Cotumaccio demonstrated that the maximum co-lex relation of an NFA $\mathcal{A}=(Q,\delta,\Sigma,s)$ can be computed in $O(\lvert \delta \rvert^{2})$ time and that it is a partial preorder over $Q$ that satisfies the requirements of Definition \ref{2:def:indpp}. In this work, we propose a new class of indexable partial preorders, which we term \textit{coarsest forward-stable co-lex orders}. Then we demonstrate that the width of our class of indexable partial preorders is always smaller than or equal to the width of the maximum co-lex relation of Cotumaccio, and in some cases, even linearly smaller with respect to the number of the automaton's states.

\section{Coarsest Forward-Stable Co-Lex Orders}\label{sec:CSF}

\paragraph{Co-Lex Relations.}
In this section we define a new category of indexable partial preorders (Definition \ref{2:def:indpp}). We start with defining co-lex relations~\cite[Definition 1]{cotumaccio2022graphs}.

\begin{definition}[Co-lex relations]\label{2:def:col_rel}
    Let $\mathcal{A}=(Q,\delta,\Sigma,s)$ be an NFA. A co-lex relation of $\mathcal{A}$ is a reflexive relation $R$ over $Q$ that satisfies:
    \begin{enumerate}
        \item For every $u, v \in Q$, with $u \neq v$, if $(u,v) \in R$, then $\lambda(u) \leq \lambda(v)$.
        \item For every pair $u \in \delta_{a}(u')$ and $v \in \delta_{a}(v')$, with $u \neq v$, if $(u,v) \in R$, then $(u',v') \in R$.
    \end{enumerate}
\end{definition}

It follows that if a co-lex relation $R$ on an NFA $\mathcal{A}=(Q,\delta,\Sigma,s)$ is also a partial order over $Q$, then $R$ is also a co-lex order of $\mathcal{A}$. Furthermore, every co-lex order of $\mathcal{A}$ is also a co-lex relation of $\mathcal{A}$. Given an NFA $\mathcal{A}$ and a co-lex relation $R$ of $\mathcal{A}$, we say that $R$ is the \textit{maximum} co-lex relation of $\mathcal{A}$ if $R$ is equal to the union of all co-lex relations of $\mathcal{A}$. Although, an NFA $\mathcal{A}$ does not always admit a maximum co-lex order, it has been proved that every NFA admits a maximum co-lex relation, denoted hereafter as $\leq_{R}$, and that $\leq_{R}$ satisfies always transitivity, i.e., $\leq_{R}$ is always a partial preorder~\cite[Lemma 5]{cotumaccio2022graphs}. From here on, given an NFA $\mathcal{A}=(Q,\delta,\Sigma,s)$ we will denote by $\mathcal{A}/_{\sim_{R}}=(Q/_{\sim_{R}},\delta/_{\sim_{R}},\Sigma,s/_{\sim_{R}})$ the quotient automaton of $\mathcal{A}$ defined by the maximum co-lex relation $\leq_{R}$, where $\sim_{R}$ is the equivalence relation over $Q$ induced by $\leq_{R}$. The next lemma shows an important property that characterizes the maximum co-lex relation of an NFA.

\begin{lemma}\cite[Corollary 16]{cotumaccio2022graphs}\label{4:lm:max_co_lex}
Let $\leq_{R}$ be the maximum co-lex relation of an automaton $\mathcal{A}=(Q,\delta,\Sigma,s)$, and let $\mathcal{A}/_{\sim_{R}}=(Q/_{\sim_{R}},\delta/_{\sim_{R}},\Sigma,s/_{\sim_{R}})$ be the quotient automaton of $\mathcal{A}$ defined by $\leq_{R}$. Then the partial order $\leq$ over $Q/_{\sim_{R}}$ induced by $\leq_{R}$ is the maximum co-lex order of $\mathcal{A}/_{\sim_{R}}$.
\end{lemma}

Lemma \ref{4:lm:max_co_lex} is important for our purposes, because it demonstrates that the quotient automaton $\mathcal{A}/_{\sim_{R}}$ always admits a maximum co-lex order. Moreover, in Lemma~\ref{dummy:lemma:1}~(1) we observe that if an NFA $\mathcal{A}$ admits a maximum co-lex order $\leq$, then $\width(\leq)$ is smaller than or equal to the width of any co-lex order of $\mathcal{A}$, since $\leq$ must be a superset of every co-lex order of $\mathcal{A}$. This in turn demonstrates that the co-lexicographic width of $\mathcal{A}$ must be equal to the width of $\leq$. In the next lemma, we demonstrate that the partition $Q/_{\sim_{R}}$ is always a forward-stable partition for $\mathcal{A}$, though we will see later that $Q/_{\sim_{R}}$ may not necessarily be the coarsest forward-stable partition for $\mathcal{A}$.

\begin{restatable}[]{lemma}{lemmaii}\label{4:lm:mx_rl_pr_fw}
    Let $\mathcal{A}=(Q,\delta,\Sigma,s)$ be an NFA and $\leq_{R}$ its maximum co-lex relation. Consider the partition $Q/_{\sim_{R}}$, where $\sim_{R}$ is the equivalence relation over $Q$ induced by $\leq_{R}$. Then $Q/_{\sim_{R}}$ is a forward-stable partition for $\mathcal{A}$.
\end{restatable}

\paragraph{Quotient Automaton $\mathcal{A}/_{\sim_{FS}}$.}

Given an NFA $\mathcal{A}$, in the following part of this article, we will show some properties of the quotient automaton $\mathcal{A}/_{\sim_{FS}}=(Q/_{\sim_{FS}},\delta/_{\sim_{FS}},\Sigma,s/_{\sim_{FS}})$, where $Q/_{\sim_{FS}}$ is the coarsest forward-stable partition for $\mathcal{A}$. In particular, in the next lemma we show that $\mathcal{A}/_{\sim_{FS}}$ always admits a maximum co-lex order.


\begin{restatable}[]{lemma}{lemmaiii}\label{4:lm:FS_max_colex}
For an NFA $\mathcal{A}=(Q,\delta,\Sigma,s)$ let $\mathcal{A}/_{\sim_{FS}}=(Q/_{\sim_{FS}},\delta/_{\sim_{FS}},\Sigma,s/_{\sim_{FS}})$ be the quotient automaton of $\mathcal{A}$ such that $Q/_{\sim_{FS}}$ is the coarsest forward-stable partition of $\mathcal{A}$ and let $\leq_{R}'$ be the maximum co-lex relation of $\mathcal{A}/_{\sim_{FS}}$. Then $\leq_{R}'$ is the maximum co-lex order of $\mathcal{A}/_{\sim_{FS}}$.   
\end{restatable}

\paragraph{CFS orders.} 
Lemma~\ref{4:lm:FS_max_colex} not only proves that $\mathcal{A}/_{\sim_{FS}}$ always admits a maximum co-lex order, but also that the maximum co-lex order of $\mathcal{A}/_{\sim_{FS}}$ is equal to the maximum co-lex relation of $\mathcal{A}/_{\sim_{FS}}$. 
At this point we have all the ingredients to introduce \textit{coarsest forward-stable co-lex (CFS) orders}.

\begin{definition}[Coarsest forward-stable co-lex (CFS) order]\label{4:def:FS_colex}
    Let $\mathcal{A}=(Q,\delta,\Sigma,s)$ be an NFA, and let $\mathcal{A}/_{\sim_{FS}}=(Q/_{\sim_{FS}},\delta/_{\sim_{FS}},\Sigma,s/_{\sim_{FS}})$ be the quotient automaton of $\mathcal{A}$ such that $Q/_{\sim_{FS}}$ is the coarsest forward-stable partition for $\mathcal{A}$. Let $\leq$ be the maximum co-lex order of $\mathcal{A}/_{\sim_{FS}}$, then we say that the \textit{coarsest forward-stable co-lex (CFS) order} of $\mathcal{A}$, denoted as $\leq_{FS}$, is the (unique) relation over $Q$ such that, for each $u,v \in Q$, $u \leq_{FS} v$ holds if and only if $[u]_{\sim_{FS}} \leq [v]_{\sim_{FS}}$.
\end{definition}

We will now demonstrate some properties of the CFS order of an NFA. 

\begin{restatable}[]{lemma}{lemmaiv}\label{4:lm:fs_pp}
    The CFS order $\leq_{FS}$ of any NFA is a partial preorder over its states.
\end{restatable}

The following remarks follow from $\leq_{FS}$ being a partial preorder.

\begin{restatable}[]{remark}{remarki}\label{4:ob:db_obs}
    For an NFA $\mathcal{A}=(Q,\delta,\Sigma,s)$ let $\mathcal{A}/_{\sim_{FS}}=(Q/_{\sim_{FS}},\delta/_{\sim_{FS}},\Sigma,s/_{\sim_{FS}})$ be the quotient automaton with $Q/_{\sim_{FS}}$ being the coarsest forward-stable partition for $\mathcal{A}$. Moreover, let $\leq_{FS}$ and $\leq$ be the CFS order of $\mathcal{A}$ and the maximum co-lex order of $\mathcal{A}/_{\sim_{FS}}$, respectively. Then (1)~$\sim_{FS}$ is equal to the equivalence relation induced by $\leq_{FS}$, and (2)~$\leq$ is the partial order induced by $\leq_{FS}$.
\end{restatable}

\section{Relation between CFS Order and Max Co-Lex Relation}
\label{sec:CSF CR}

At this point, we have introduced the maximum co-lex relation $\leq_{R}$ and the CFS order $\leq_{FS}$ of an NFA. In this section we study the relation between $\leq_{R}$ and $\leq_{FS}$. In particular, we can prove that $\leq_{FS}$ is always a superset of $\leq_{R}$.
\begin{restatable}[]{lemma}{lemmavii}\label{4:lm:fs_rf_cr}
    For an NFA $\mathcal{A} = (Q,\delta,\Sigma,s)$ let $\leq_{R}$ be the maximum co-lex relation of $\mathcal{A}$ and let $\leq_{FS}$ be the CFS order of $\mathcal{A}$. Then $\leq_{FS}$ is a superset of $\leq_{R}$.
\end{restatable}

The following lemma then allows us to conclude that the width of the CFS order is always at most the width of the maximum co-lex relation.

\begin{restatable}[]{lemma}{dummylemmai}\label{dummy:lemma:1}
    The following two statements hold:
    (1) Let $\leq$ and $\leq'$ be two partial orders or partial preorders over the same set $U$. If $\leq$ is a superset of $\leq'$, then $\width(\leq) \leq \width(\leq')$. (2) Let $\leq$ be a partial preorder over a set $U$ and let $\leq'$ be the partial order induced by $\leq$ over $U/_{\sim}$. Then, $\width(\leq) = \width(\leq')$.
\end{restatable}

We are now ready to prove that CFS orders are a class of indexable partial preorders, see Definition \ref{2:def:indpp}. Moreover, we will show that the quotient automaton $\mathcal{A}/_{\sim_{FS}}$ has never more states than $\mathcal{A}/_{\sim_{R}}$, and that the co-lexicographic width of the former is never larger then the co-lexicographic width of the latter.

\begin{restatable}[]{theorem}{mainthi}\label{main:th:1}
    Let $\mathcal{A}=(Q,\delta,\Sigma,s)$ be an NFA, and let $\leq_{R}$ and $\leq_{FS}$ be the maximum co-lex relation and the CFS order of $\mathcal{A}$, respectively. Finally, let $\mathcal{A}/_{\sim_{R}}=(Q/_{\sim_{R}},\delta/_{\sim_{R}},\Sigma,s/_{\sim_{R}})$ and $\mathcal{A}/_{\sim_{FS}}=(Q/_{\sim_{FS}},\delta/_{\sim_{FS}},\Sigma,s/_{\sim_{FS}})$ be the quotient automata of $\mathcal{A}$ defined by $\leq_{R}$ and $\leq_{FS}$, respectively. Then:
    \begin{enumerate}
        \item $\leq_{FS}$ is an indexable partial preorder for $\mathcal{A}$.
        \item The co-lexicographic width of $\mathcal{A}/_{\sim_{FS}}$ is smaller than or equal to the co-lexicographic width of $\mathcal{A}/_{\sim_{R}}$.
        \item The number of parts in partition $Q/_{\sim_{FS}}$ is smaller than or equal to the number of parts in partition $Q/_{\sim_{R}}$. 
    \end{enumerate}
\end{restatable}
\begin{proof}
\begin{enumerate}
    \item By Lemma \ref{4:lm:fs_pp}, $\leq_{FS}$ is a partial preorder over $Q$, thus it remains to prove the three properties in Definition~\ref{2:def:indpp}. 
    (a)~Clearly, $\mathcal{A}$ and $\mathcal{A}/_{\sim_{FS}}$ accept the same language, since states in the same part of a forward-stable partition are reached by the same set of strings~\cite[Lemma 4.7]{alanko2021wheeler}.
    (b)~Let $\bar{p}$ be the co-lexicographic width of $\mathcal{A}$. We have to prove that the co-lexicographic width of $\mathcal{A}/_{\sim_{FS}}$ is smaller than or equal to $\bar{p}$. 
    Let $\leq$ be the maximum co-lex order of $\mathcal{A}/_{\sim_{FS}}$. Remark~\ref{4:ob:db_obs} yields that $\leq$ is the partial order induced by $\leq_{FS}$ and hence Lemma~\ref{dummy:lemma:1}~(2) implies that $\width(\le) = \width(\le_{FS})$. Then, Lemma~\ref{4:lm:fs_rf_cr} yields that $\leq_{FS}$ is a superset of $\le_{R}$, the maximum co-lex relation of $\mathcal{A}$, and thus Lemma~\ref{dummy:lemma:1}~(1) implies that $\width(\le_{FS}) \le \width(\le_R)$. Now let $\leq'$ be a co-lex order of $\mathcal{A}$ of width $\bar{p}$. Then, $\leq'$ is also a co-lex relation and thus $\leq_{R}$ is a superset of $\leq'$. Applying Lemma~\ref{dummy:lemma:1}~(1) again yields $\width(\leq_{R})\le \width(\leq')=\bar p$. Thus altogether
    $
        \width(\le) 
        \le \bar p
    $.
    (c)~Let $\bar{q}$ be the co-lexicographic width of $\mathcal{A}/_{\sim_{FS}}$, and let $\leq$ be the maximum co-lex order of $\mathcal{A}/_{\sim_{FS}}$. Remark~\ref{4:ob:db_obs} yields that $\leq$ is the partial order induced by $\leq_{FS}$. Hence, we have to prove that $\width(\leq)=\bar{q}$. It is clear that $\width(\leq)\ge \bar q$. For the other direction, let $\le'$ be a co-lex order of  $\mathcal{A}/_{\sim_{FS}}$ with $\width(\leq')= \bar q$. Since $\leq$ is a superset of $\leq'$, Lemma~\ref{dummy:lemma:1}~(2) yields $\width(\leq)\le\width(\le')=\bar{q}$.

    \item Remark~\ref{4:ob:db_obs} yields that the maximum co-lex order of $\mathcal{A}/_{\sim_{FS}}$ is the partial order induced by $\leq_{FS}$. Hence the co-lexicographic width of $\mathcal{A}/_{\sim_{FS}}$ is equal to $\width(\leq_{FS})$ by Lemma \ref{dummy:lemma:1}~(2). Lemma \ref{4:lm:fs_rf_cr} states that $\le_{FS}$ is a superset of $\le_R$ and thus $\width(\leq_{FS}) \le \width(\leq_R)$ according to Lemma~\ref{dummy:lemma:1}~(1). From Lemma~\ref{4:lm:max_co_lex}, we conclude that the partial order $\le$ induced by $\leq_{R}$ is the maximum co-lex order of $\mathcal A/_{\sim_R}$. Again using Lemma~\ref{dummy:lemma:1}~(2) implies $\width(\leq_{R})=\width(\le)$.
    
    \item By Lemma \ref{4:lm:mx_rl_pr_fw}, $Q/_{\sim_{R}}$ is forward-stable for $\mathcal{A}$, while $Q/_{\sim_{FS}}$ is the coarsest forward-stable partition for $\mathcal{A}$. Thus $Q/_{\sim_{R}}$ is a refinement of $Q/_{\sim_{FS}}$ and the number of parts of $Q/_{\sim_{FS}}$ is at most the number of parts of $Q/_{\sim_{R}}$.\qedhere
\end{enumerate}
\end{proof}

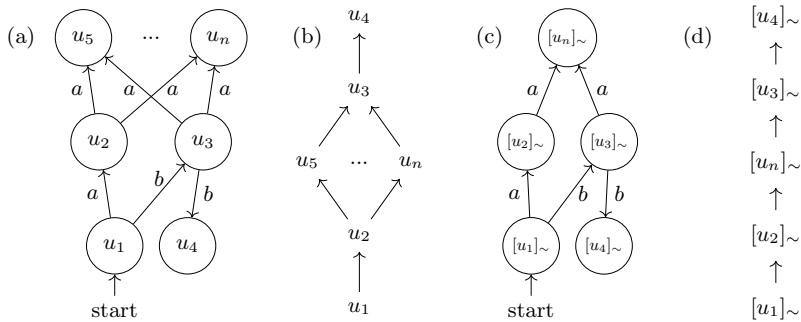
\begin{figure}[ht!]
\begin{center}
\resizebox{0.85\textwidth}{!}{
\begin{tikzpicture}
[dim/.style={minimum size=2.5em}, dots/.style={text centered}, 
scale=0.8, dim2/.style={minimum size=1.4em}, dim3/.style={minimum size=2.2em},
dim4/.style={minimum size=2.8, scale=0.75}]
    \begin{scope}
        \node[dots] (-1) at (0.5, 0.0) {(a)};
        
        \node[initial,state,initial where=below,dim] (1) at (2.3,-4) {$u_{1}$};
        \node[state,dim] (2) at (2.0,-2.0) {$u_{2}$};
        \node[state,dim] (3) at (4.0,-2.0) {$u_{3}$};
        \node[state,dim] (4) at (3.7, -4) {$u_{4}$};
        \node[state,dim] (5) at (1.7,0) {$u_{5}$};
        \node[state,dim] (6) at (4.3,0) {$u_{n}$};
        \node[dots]  (0) at (3.0,0) {...};
    
        \path [->] (1) edge node[left] {$a$} (2);
        \path [->] (1) edge node[above] {$b$} (3);
        \path [->] (2) edge node[left] {$a$} (5);
        \path [->] (2) edge node[right] {$a$} (6);
        \path [->] (3) edge node[left] {$a$} (5);
        \path [->] (3) edge node[right] {$a$} (6);
        \path [->] (3) edge node[right] {$b$} (4);
    \end{scope}

    \begin{scope}[xshift=4cm]
        \node[dots] (-1) at (2.0, 0.0) {(b)};
        \node[dim2] (1) at (3,-5.2) {$u_{1}$};
        \node[dim2] (2) at (3,-3.8) {$u_{2}$};
        \node[dim2] (5) at (2.0,-2.4) {$u_{5}$};
        \node[dim2] (6) at (4.0,-2.4) {$u_{n}$};
        \node[dim2] (3) at (3,-1.0) {$u_{3}$};
        \node[dim2] (4) at (3,0.4) {$u_{4}$};
        \node[dots]  (0) at (3.0,-2.4) {...};
        \node[dots] (-2) at (1.6, 0.0) {};
    
        \path [->] (1) edge node[left] {} (2);
        \path [->] (2) edge node[left] {} (5);
        \path [->] (2) edge node[left] {} (6);
        \path [->] (6) edge node[left] {} (3);
        \path [->] (5) edge node[left] {} (3);
        \path [->] (3) edge node[right] {} (4);
    \end{scope}

    \begin{scope}[xshift=8cm]
        \node[dots] (-1) at (1.5, 0.0) {(c)};
        \node[initial,state,initial where=below,dim4] (1) at (2.3,-4) {$[u_{1}]_{\sim}$};
        \node[state, dim4] (2) at (2.2,-2) {$[u_{2}]_{\sim}$};
        \node[state, dim4] (3) at (3.8,-2) {$[u_{3}]_{\sim}$};
        \node[state, dim4] (4) at (3.7,-4) {$[u_{4}]_{\sim}$};
        \node[state, dim4] (5) at (3,0) {$[u_{n}]_{\sim}$};
        \node[dots] (-2) at (1.2, 0.0) {};
    
        \path [->] (1) edge node[left] {$a$} (2);
        \path [->] (1) edge node[right] {$b$} (3);
        \path [->] (2) edge node[left] {$a$} (5);
        \path [->] (3) edge node[right] {$a$} (5);
        \path [->] (3) edge node[right] {$b$} (4);
    \end{scope}
    \begin{scope}[xshift=12cm]
        \node[dots] (-1) at (1.5, 0.0) {(d)};
        \node[dim3] (1) at (3,-5.2) {$[u_{1}]_{\sim}$};
        \node[dim3] (2) at (3,-3.8) {$[u_{2}]_{\sim}$};
        \node[dim3] (5) at (3,-2.4) {$[u_{n}]_{\sim}$};
        \node[dim3] (3) at (3,-1.0) {$[u_{3}]_{\sim}$};
        \node[dim3] (4) at (3.0,0.4) {$[u_{4}]_{\sim}$};
        \node[dots] (-2) at (1.7, 0.0) {};
    
        \path [->] (1) edge node[left] {} (2);
        \path [->] (2) edge node[left] {} (5);
        \path [->] (5) edge node[left] {} (3);
        \path [->] (3) edge node[left] {} (4);
    \end{scope}
\end{tikzpicture}
}
\end{center}
\caption{(a) An NFA $\mathcal{A}=(Q,\delta,\Sigma,s)$ with $\Sigma = \{a,b\}$,
$Q=\{u_{1},...,u_{n}\}$, with $n > 4$, where $u_{1}=s$ is the initial state, and $\delta$ is s.t.\ $u_{2} \in \delta_{a}(u_{1})$, $u_{3} \in \delta_{b}(u_{1})$, $u_{4}\ \in \delta_{b}(u_{3})$ and for each $4 < i \leq n$, $u_{i} \in \delta_{a}(u_{2})$, $u_{i} \in \delta_{a}(u_{3})$. We observe that $\mathcal{A}/_{\sim_{R}}$ is equal to $\mathcal{A}$ itself. Note that the example is not trivial, as states $u_{2}$ and $u_{3}$ may not be merged without changing the language of $\mathcal{A}$ (due to state $u_{4}$). (b) The Hasse diagram of the maximum co-lex order of $\mathcal{A}/_{\sim_{R}}$, where $\{u_{5},...,u_{n}\}$ forms a largest antichain, consequently the co-lexicographic width of $\mathcal{A}/_{\sim_{R}}$ is equal to $n-4$. (c) The automaton $\mathcal{A}/_{\sim_{FS}}$ consisting of five states, where for each $4 < i \leq n$, $u_{i} \in [u_{n}]_{\sim_{FS}}$. (d) The Hasse diagram of the maximum co-lex order of $\mathcal{A}/_{\sim_{FS}}$. Note that this total order is also a Wheeler order of $\mathcal{A}/_{\sim_{FS}}$, i.e., the co-lexicographic width of $\mathcal{A}/_{\sim_{FS}}$ is equal to 1.}
\label{4:fg:ex1}
\end{figure}

We now prove another important result regarding the automata $\mathcal{A}/_{\sim_{FS}}$ and $\mathcal{A}/_{\sim_{R}}$. Namely, that the number of states and the co-lexicographic width of $\mathcal{A}/_{\sim_{FS}}$ can be in some cases arbitrary smaller with respect to those of $\mathcal{A}/_{\sim_{R}}$.

\begin{restatable}[]{theorem}{mainthii}\label{main:th:2}
    There exists a class of NFAs such that, for each NFA $\mathcal{A}=(Q,\delta,\Sigma,s)$ in this class; the quotient automaton $A/_{\sim_{R}}$, defined by the maximum co-lex relation of $\mathcal{A}$, has a number of states and a co-lexicographic width equal to $O(\lvert Q \rvert)$, while the quotient automaton $A/_{\sim_{FS}}$, defined by the CFS order of $\mathcal{A}$, has a number of states and a co-lexicographic width equal to $O(1)$.
\end{restatable}
\begin{proof}
    We describe the structure of this class of NFAs in Figure \ref{4:fg:ex1}. In particular, in this figure we can observe that the automaton $\mathcal{A}/_{\sim_{FS}}$ has always co-lexicographic width equal to 1 and it is a Wheeler NFA.
\end{proof}

Finally, we show that the CFS order of an automaton $\mathcal{A}=(Q,\delta,\Sigma,s)$  can be computed in $O(\lvert \delta \rvert^{2})$ time.

\begin{corollary}\label{unique:cor}
    Let $\mathcal{A}=(Q,\delta,\Sigma,s)$ be an NFA. We can compute the CFS order $\leq_{FS}$ of $\mathcal{A}$ in $O(\lvert \delta \rvert^{2})$ time.
\end{corollary}

\begin{proof}
    It is possible to use the partition refinement framework of Paige and Tarjian \cite{paige1987three} to compute the coarsest forward-stable partition of $\mathcal{A}$ and consequently the quotient automaton $\mathcal{A}/_{\sim_{FS}}=(Q/_{\sim_{FS}},\delta/_{\sim_{FS}},\Sigma,s_{\sim_{FS}})$ in $O(\lvert \delta \rvert \log \lvert Q \rvert)$ time. After that, we can compute the maximum co-lex relation $\leq_{R}'$ of $\mathcal{A}/_{\sim_{FS}}$ using Cotumaccio's algorithm \cite[Theorem 8]{cotumaccio2022graphs} in $O(\lvert \delta \rvert^{2})$ time, where by Lemma \ref{4:lm:FS_max_colex} we know that $\leq_{R}'$ is also the maximum co-lex order of $\mathcal{A}/_{\sim_{FS}}$. Once we have computed $\leq_{R}'$, we can also reconstruct $\leq_{FS}$.
\end{proof}

\begin{credits}
\subsubsection{\ackname} We would like to thank Nicola Cotumaccio for insightful discussions on the topic of this paper. Ruben Becker, Sung-Hwan Kim, Nicola Prezza, and Carlo Tosoni are funded by the European Union (ERC, REGINDEX, 101039208). Views and opinions expressed are however those of the author(s) only and do not necessarily reflect those of the European Union or the European Research Council. Neither the European Union nor the granting authority can be held responsible for them. 

\end{credits}

\newpage 
\bibliographystyle{splncs04}
\bibliography{Chapters/bibliography}

\newpage
\appendix
\section{Deferred Material}\label{appendix: deferred material}
\subsection{Deferred Proofs for Section~\ref{sec:CSF}}
We give the defered proofs for Section~\ref{sec:CSF}.
\lemmaii*
\begin{proof}
    Let $S, T$ be two arbitrary parts of partition $Q/_{\sim_{R}}$. We want to demonstrate that for each $a \in \Sigma$ either $S \subseteq \delta_{a}(T)$ or $S \cap \delta_{a}(T) = \emptyset$ holds. Let $a \in \Sigma$ and suppose for contradiction that there exists $u \in S \cap \delta_{a}(T)$ and $v \in S \setminus \delta_{a}(T)$. Then there exists $u' \in T$ such that $u \in \delta_{a}(u')$ and $S = [u]_{\sim_{R}}$, $T = [u']_{\sim_{R}}$. Since $v$ and $u$ belong to the same part $S$ of $Q/_{\sim_{R}}$, it must hold that $v \leq_{R} u$ and $u \leq_{R} v$. It implies that, due to Axiom 1 of Definition \ref{2:def:col_rel}, $\lambda(v) \leq \lambda(u)$ and $\lambda(u) \leq \lambda(v)$, which can simultaneously hold if and only if $\lambda(v) = \lambda(u) = \{a\}$. It follows that there exists $v' \in Q$ such that $v \in \delta_{a}(v')$. Due to Axiom 2 of Definition \ref{2:def:col_rel}, since $v \leq_{R} u$ we have that $v' \leq_{R} u'$, moreover, since $u \leq_{R} v$ we have that $u' \leq_{R} v'$ and so $u' \sim_{R} v'$. Therefore $v' \in [u']_{\sim_{R}}$ and so $v' \in T$. Therefore, it follows that $v \in \delta_{a}(T)$, a contradiction.
\end{proof}

\lemmaiii*
\begin{proof}
    As the set of co-lex relations is a superset of the set of co-lex orders, it suffices to show that $\le_R'$ is a co-lex order. By Definition $\le_R'$ is reflexive and by maximality $\le_R'$ is transitive~\cite[Lemma 5]{cotumaccio2022graphs}. Hence, it remains to show antisymmetry.
    
    Therefore, let us suppose for the purpose of contradiction that $\le_R'$ is not antisymmetric, and let us consider the equivalence relation $\sim_{R}'$ over $Q/_{\sim_{FS}}$ induced by $\leq_{R}'$. Moreover, let us consider the partition $\mathcal{P}$ of $Q$ such that states $u,v \in Q$ belong to the same part of $\mathcal{P}$ if and only if $[u]_{\sim_{FS}} \sim_{R}' [v]_{\sim_{FS}}$ holds. Clearly, partition $Q/_{\sim_{FS}}$ is a refinement of $\mathcal{P}$ since every part of $Q/_{\sim_{FS}}$ is contained in a part of $\mathcal{P}$. However, as $\leq_{R}'$ is not antisymmetric, it is easy to observe that the reverse does not hold, and therefore $\mathcal{P}$ is not a refinement of $Q/_{\sim_{FS}}$. 
    
    At this point, we want to prove that $\mathcal{P}$ is a forward-stable partition for $\mathcal{A}$. Let $S, T$ be two arbitrary parts of $\mathcal{P}$. We have to prove that $\forall a \in \Sigma$, $S \subseteq \delta_{a}(T)$ or $S \cap \delta_{a}(T) = \emptyset$ hold. Let $a \in \Sigma$ and suppose for contradiction that there exists $u \in \delta_{a}(T) \cap S$, and $v \in S \setminus \delta_{a}(T)$, so we know that there exists $u' \in T$, such that $u \in \delta_{a}(u')$. Clearly, $v \notin [u]_{\sim_{FS}}$, since if $v \in [u]_{\sim_{FS}}$, then, due to the fact that $Q/_{\sim_{FS}}$ is forward-stable, it must exist a state $v' \in [u']_{\sim_{FS}}$, such that $v \in \delta_{a}(v')$. Therefore we have $[v]_{\sim_{FS}} \neq [u]_{\sim_{FS}}$. Since $[v]_{\sim_{FS}} \sim_{R}' [u]_{\sim_{FS}}$, and due to Axiom 1 of Definition \ref{2:def:col_rel}, we have that $\lambda([v]_{\sim_{FS}})=\lambda([u]_{\sim_{FS}})=\{a\}$. It follows that there exists a state $[v']_{\sim_{FS}} \in Q/_{\sim_{FS}}$ such that $[v]_{\sim_{FS}} \in \delta/_{\sim_{FS}}([v']_{\sim_{FS}}, a)$. Due to Axiom 2 of Definition \ref{2:def:col_rel}, since $[u]_{\sim_{FS}} \leq_{R}' [v]_{\sim_{FS}}$, it follows that $[u']_{\sim_{FS}} \leq_{R}' [v']_{\sim_{FS}}$, and since $[v]_{\sim_{FS}} \leq_{R}' [u]_{\sim_{FS}}$, it follows that $[v']_{\sim_{FS}} \leq_{R}' [u']_{\sim_{FS}}$; therefore $[v']_{\sim_{FS}} \sim_{R}' [u']_{\sim_{FS}}$ holds.
    Due to the fact that $Q/_{\sim_{FS}}$ is forward-stable, there exists $v'' \in [v']_{\sim_{FS}}$ such that $v \in \delta_{a}(v'')$, moreover, since $[u']_{\sim_{FS}} \sim_{R}' [v'']_{\sim_{FS}}$, it follows that $v'' \in T$, which in turn proves that $v \in \delta_{a'}(T)$. Therefore $\mathcal{P}$ is a forward-stable partition for $\mathcal{A}$; however, $\mathcal{P}$ is not a refinement of $Q/_{\sim_{FS}}$, this contradicts the hypothesis that $Q/_{\sim_{FS}}$ is the coarsest forward-stable partition of $\mathcal{A}$. It follows that it was absurd to assume that $\leq_{R}'$ was not antisymmetric.
\end{proof}

\lemmaiv*

\begin{proof}
    We have to prove that $\leq_{FS}$ satisfies reflexivity and transitivity. Let us consider the maximum co-lex order $\leq$ of the quotient automaton $\mathcal{A}/_{\sim_{FS}}=(Q/_{\sim_{FS}},\delta/_{\sim_{FS}},\Sigma,s/_{\sim_{FS}})$. By definition of co-lex order, $\leq$ is a partial order, therefore it satisfies reflexivity and transitivity.
    Let us consider an arbitrary state $u \in Q$. By Definition \ref{4:def:FS_colex}, $u \leq_{FS} u$ if and only if $[u]_{\sim_{FS}} \leq [u]_{\sim_{FS}}$; however $[u]_{\sim_{FS}} \leq [u]_{\sim_{FS}}$ holds since $\leq$ is reflexive, therefore also $u \leq_{FS} u$ must hold. It follows that $\leq_{FS}$ satisfies reflexivity.
    Let us consider three arbitrary states $u,v,z \in Q$ such that $u \leq_{FS} v$ and $v \leq_{FS} z$. We have to prove that $u \leq_{FS} z$. Since $u \leq_{FS} v$ we have that $[u]_{\sim_{FS}} \leq [v]_{\sim_{FS}}$, and since $v \leq_{FS} z$ we have that $[v]_{\sim_{FS}} \leq [z]_{\sim_{FS}}$. Moreover, due to the fact that $\leq$ is transitive, we have also that $[u]_{\sim_{FS}} \leq [z]_{\sim_{FS}}$, and consequently $u \leq_{FS} z$. It follows that $\leq_{FS}$ satisfies transitivity.
\end{proof}

\remarki*
\begin{proof}
    For~(1), observe that $\leq$ is a partial order and thus satisfies antisymmetry. Therefore, given two states $u,v \in Q$, $u \leq_{FS} v$ and $v \leq_{FS} u$ can both hold if and only if $u$ and $v$ belong to the same part of the coarsest forward-stable partition $Q/_{\sim_{FS}}$ for $\mathcal{A}$, i.e., if and only if $u \sim_{FS} v$.
    For~(2), note that $\leq_{FS}$ is a partial preorder over $Q$, while $\leq$ is a partial order over $Q/_{\sim_{FS}}$. Moreover, $u \leq_{FS} v$ holds if and only if $[u]_{\sim_{FS}} \leq [v]_{\sim_{FS}}$ holds.
\end{proof}

\subsection{Deferred Proofs for Section~\ref{sec:CSF CR}}
Before proving Lemma~\ref{4:lm:fs_rf_cr}, we have to repeat some concepts from the work of Cotumaccio, starting from preceding pairs~\cite[Definition 6]{cotumaccio2022graphs}.

\begin{definition}[Preceding pairs]\label{4:def:prec_pair}
    Let $\mathcal{A}=(Q,\delta,\Sigma,s)$ be an NFA and let $(u',v'),(u,v) \in Q \times Q$ be pairs of distinct states. The pair $(u',v')$ precedes $(u,v)$ in $\mathcal{A}$ if there are $u_{1},...,u_{r}, v_{1},...,v_{r} \in Q$ with $r \geq 1$ and $a_{1}, ..., a_{r-1} \in \Sigma$ s.t.\ 
    (1)~$u_{1} = u'$ and $v_{1} = v'$, 
    (2)~$u_{r} = u$ and $v_{r} = v$,
    (3)~for all $i = 1, \ldots , r-1$, $u_{i} \neq v_{i}$, and 
    (4)~for all $i = 1, \ldots , r-1$, $u_{i+1} \in \delta_{a_{i}}(u_{i})$ and $v_{i+1} \in \delta_{a_{i}}(v_{i})$.
\end{definition}

Cotumaccio characterized the maximum co-lex relation in terms of preceding pairs by showing that $u\leq_{R} v$ holds if and only if for all pairs $(u',v')$ preceding $(u,v)$ it holds that $\lambda(u') \leq \lambda(v')$~\cite[Lemma 7]{cotumaccio2022graphs}.
We have the following lemma regarding preceding pairs in quotient automata.
\begin{restatable}[]{lemma}{lemmavi}\label{4:lemma:prec_pair_qa}
    Let $\mathcal{A}=(Q,\delta,\Sigma,s)$ be an NFA, and let $\mathcal{A}/_{\sim}=(Q/_{\sim},\delta/_{\sim},\Sigma,s/_{\sim})$ be a quotient automaton of $\mathcal{A}$ such that $Q/_{\sim}$ is a forward-stable partition for $\mathcal{A}$. Consider two pairs $([u]_{\sim},[v]_{\sim}), ([u']_{\sim},[v']_{\sim})$, with $[u]_{\sim},[v]_{\sim},[u']_{\sim},[v']_{\sim} \in Q/_{\sim}$, such that $([u']_{\sim},[v']_{\sim})$ precedes $([u]_{\sim},[v]_{\sim})$ in $\mathcal{A}/_{\sim}$. Then, there exist $u'',v'' \in Q$, with $u'' \in [u']_{\sim}, v'' \in [v']_{\sim}$, such that $(u'',v'')$ precedes $(u,v)$ in $\mathcal{A}$.
\end{restatable}
\begin{proof}
        Since $([u']_{\sim},[v']_{\sim})$ precedes $([u]_{\sim},[v]_{\sim})$ in $\mathcal{A}/_{\sim}$, we know that there exists a sequence $[u_{1}]_{\sim},...,[u_{r}]_{\sim},[v_{1}]_{\sim},...,[v_{r}]_{\sim}$ that satisfies the requirements of Definition \ref{4:def:prec_pair}. The proof proceeds by induction over the parameter $r$. Firstly, let us prove the statement for $r=1$. We know that $u,v$ are distinct states, because $[u]_{\sim}, [v]_{\sim}$ are distinct, and therefore pair $(u,v)$ trivially precedes itself in $\mathcal{A}$. Now let us prove the statement for a general $r$. Let us consider the sequence $[u_{2}]_{\sim},...,[u_{r}]_{\sim},[v_{2}]_{\sim},...,[v_{r}]_{\sim}$, due to the inductive hypothesis, we know that there exist $u^{*} \in [u_{2}]_{\sim}$ and $v^{*} \in [v_{2}]_{\sim}$, such that $(u^{*},v^{*})$ precedes $(u,v)$ in $\mathcal{A}$. Since $[u_{2}]_{\sim} \in \delta/_{\sim}([u_{1}]_{\sim}, a_{1})$ and $[v_{2}]_{\sim} \in \delta/_{\sim}([v_{1}], a_{1})$, and due to the definition of quotient automaton, we know that there exist $\bar{u}_{1},\bar{u}_{2},\bar{v}_{1},\bar{v}_{2} \in Q$ such that: $\bar{u}_{1} \in [u_{1}]_{\sim}$, $\bar{u}_{2} \in [u_{2}]_{\sim}$, $\bar{v}_{1} \in [v_{1}]_{\sim}$, $\bar{v}_{2} \in [v_{2}]_{\sim}$, $\bar{u_{2}} \in \delta(\bar{u_{1}}, a_{1})$, $\bar{v_{2}} \in \delta(\bar{v_{1}}, a_{1})$. Moreover, since $Q/_{\sim}$ is forward-stable, there must exist also $u'' \in [u_{1}]_{\sim}$ and $v'' \in [v_{1}]_{\sim}$, such that $u^{*} \in \delta(u'',a_{1})$ and $v^{*} \in \delta(v'',a_{1})$ (note that $u'' \neq v''$ because $[u_{1}]_{\sim} \neq [v_{1}]_{\sim}$). Moreover, since $(u^{*},v^{*})$ precedes $(u,v)$ in $\mathcal{A}$, then also $(u'',v'')$ must precede $(u,v)$ in $\mathcal{A}$. It follows that the lemma holds.
\end{proof}
We are now ready to prove Lemma~\ref{4:lm:fs_rf_cr}.
\lemmavii*
\begin{proof}
    We suppose for contradiction that $\leq_{FS}$ is not a superset of $\leq_{R}$. Then, there is $(u,v) \in \; \leq_{R}$ with $(u,v) \notin \; \leq_{FS}$. Let $\mathcal{A}/_{\sim_{FS}}=(Q/_{\sim_{FS}},\delta/_{\sim_{FS}},\Sigma,s/_{\sim_{FS}})$ be the quotient automaton of $\mathcal{A}$ such that $Q/_{\sim_{FS}}$ is the coarsest forward-stable partition for $\mathcal{A}$, and let $\leq$ be the maximum co-lex order of $\mathcal{A}/_{\sim_{FS}}$ (which exists due to Lemma \ref{4:lm:FS_max_colex}). By Definition \ref{4:def:FS_colex}, $(u,v) \notin\; \leq_{FS}$ if and only if $([u]_{\sim_{FS}}, [v]_{\sim_{FS}}) \notin\; \leq$. In addition, $[u]_{\sim_{FS}} \neq [v]_{\sim_{FS}}$, since $\leq$ satisfies reflexivity. Due to Lemma 7 in the article by Cotumaccio~\cite{cotumaccio2022graphs}
    and because $\leq$ is also the maximum co-lex relation of $\mathcal{A}/_{\sim_{FS}}$ (see Lemma \ref{4:lm:FS_max_colex}), there exist $[u']_{\sim_{FS}}, [v']_{\sim_{FS}} \in Q/_{\sim_{FS}}$ such that the pair $([u']_{\sim_{FS}}, [v']_{\sim_{FS}})$ precedes $([u]_{\sim_{FS}}, [v]_{\sim_{FS}})$ in $\mathcal{A}/_{\sim_{FS}}$, and $\lambda([u']_{\sim_{FS}}) \leq \lambda([v']_{\sim_{FS}})$ does not hold. Due to Lemma \ref{4:lemma:prec_pair_qa}, there exist $u'' \in [u']_{\sim_{FS}}$, $v'' \in [v']_{\sim_{FS}}$ such that $(u'',v'')$ precedes $(u,v)$ in $\mathcal{A}$. Moreover, since $u'' \in [u']_{\sim_{FS}}$, we have that $\lambda(u'') = \lambda([u']_{\sim_{FS}})$, and analogously  $\lambda(v'') = \lambda([v']_{\sim_{FS}})$. Hence, if $\lambda([u']_{\sim_{FS}}) \leq \lambda([v']_{\sim_{FS}})$ does not hold, then also $\lambda(u'') \leq \lambda(v'')$ does not hold. However, again using Lemma 7 in the article by Cotumaccio~\cite{cotumaccio2022graphs}, this contradicts the assumption $(u,v) \in\;\leq_{R}$. Hence, $\leq_{FS}$ is a superset of $\leq_{R}$.
\end{proof}

\dummylemmai*
\begin{proof}
    (1) Since $\leq$ is a superset of $\leq'$, we know that for each $u,v \in U$, if $(u,v) \notin\, \leq$, then $(u,v) \notin\,\leq'$. It follows that every antichain according to $\leq$ is also an antichain according to $\leq'$. (2) Let us denote with $p$ and $p'$ the integers $\width(\leq)$ and $\width(\leq')$, respectively. Let us consider $L$, an arbitrary antichain according to $\leq$, and let us define the set $L' \coloneqq \{ [u]_{\sim} : u \in L \}$; clearly $\lvert L \rvert = \lvert L' \rvert$. It is easy to observe that $L'$ is an antichain according to $\leq'$, this proves that $p \leq p'$. Now, let us consider $L'$, an arbitrary antichain according to $\leq'$, and let us consider a possible set $L$, in such a way that $L$ contains a representative of each equivalence class $[u]_{\sim}$, such that $[u]_{\sim} \in L'$. Clearly, $L$ is an antichain according to $\leq$, and $\lvert L' \rvert = \lvert L \rvert$. This proves that $p' \leq p$. It follows that $p = p'$.
\end{proof}

\section{Relationships between Orders in Figure~\ref{fg:ex:3}}
\label{appendix: relationships}


In this section we prove some properties regarding the different relations presented in Figure~\ref{fg:ex:3}, namely \textit{Wheeler orders}, \textit{Wheeler preorders}, \textit{(maximum) co-lex orders}, \textit{(maximum) co-lex relations}, and \textit{coarsest forward-stable co-lex orders}.

\begin{proposition}
    Let $\mathcal{A}$ be an NFA. Then, the following statements hold true: 
    \begin{enumerate}
        \item If $\mathcal{A}$ admits a maximum co-lex order $\leq$, then the coarsest forward-stable co-lex order of $\mathcal{A}$ may not be equal to $\leq$.
        \item The maximum co-lex relation of $\mathcal{A}$ may not be equal to the coarsest forward-stable co-lex order of $\mathcal{A}$.
        \item If $\mathcal{A}$ admits a Wheeler preorder $\leq$ and a Wheeler order $\leq'$, then $\leq'$ may be not equal to $\leq$.
    \end{enumerate}
\end{proposition}

\begin{proof}
    \begin{enumerate}
        \item An example is shown in Figure \ref{4:fg:ex1}.
        \item An example is shown in Figure \ref{4:fg:ex1}.
        \item Let us consider the NFA $\mathcal{A}=(Q,\delta,\Sigma,s)$ with the following structure: $Q = \{u_{1}, u_{2}, u_{3}\}$, $\Sigma = \{a\}$, $s = u_{1}$ and $u_{2}, u_{3} \in \delta_{a}(u_{1})$. It is possible to check that $\leq\, \coloneqq \{(u_{1}, u_{1}), (u_{1}, u_{2}), (u_{1}, u_{3}), (u_{2}, u_{2}), (u_{2}, u_{3}), (u_{3}, u_{3}) \}$ is a Wheeler order of $\mathcal{A}$. However, 
        \[
            \leq'\,\coloneqq \{(u_{1}, u_{1}), (u_{1}, u_{2}), (u_{1}, u_{3}), (u_{2}, u_{2}), (u_{2}, u_{3}), (u_{3},u_{2}), (u_{3}, u_{3}) \}
        \]
        is a Wheeler preorder of $\mathcal{A}$, and $\leq'\;\neq\;\leq$. \qedhere
    \end{enumerate}
\end{proof}

The following proposition characterizes the case when the maximum co-lex relation and the coarsest forward-stable co-lex order are identical.

\begin{proposition} \label{thm:relationship}
    Let $\mathcal{A}$ be an NFA. Let $\le_{R}$ and $\le_{FS}$ be the maximum co-lex relation and the coarsest forward-stable co-lex order of $\mathcal{A}$, respectively. Then, the following statements hold true:
    \begin{enumerate}
        \item \label{thm:relationship:r:fs} $\le_R$ is equal to $\le_{FS}$ if and only if $Q/_{\sim_{R}}=Q/_{\sim_{FS}}$.
        \item \label{thm:relationship:o} Assume that $\mathcal{A}$ admits a maximum co-lex order $\leq$, then,
        \begin{enumerate}
            \item \label{thm:relationship:o:r} $\le_{R}$ is equal to $\leq$ if and only if $|Q|=|Q/_{\sim R}|$.
            \item \label{thm:relationship:o:fs }$\le_{FS}$ is equal to $\leq$ if and only if $|Q|=|Q/_{\sim FS}|$. \qedhere
        \end{enumerate}
    \end{enumerate}
\end{proposition}
\begin{proof}
    \begin{enumerate}
        \item $(\Rightarrow)$: Note that if, for all $u,v\in Q$, $u\le_{R}~v$ if and only if $u\le_{FS} v$, then it also holds that, for all $u,v\in Q$, $u\sim_{R} v$ if and only if $u\sim_{FS} v$. Consequently, we obtain $Q/_{\sim_{R}}= Q/_{\sim_{FS}}$.

        $(\Leftarrow)$: Assume $Q/_{\sim_{R}}=Q/_{\sim_{FS}}$. We need to prove that, for every $u,v\in Q$, $u\le_{R} v$ holds if and only if $u\le_{FS}v$ holds. The forward direction immediately follows from Lemma~\ref{4:lm:fs_rf_cr}, thus we shall prove that $u\le_{FS}v$ implies $u\le_{R} v$. Suppose for the purpose of contradiction that this is not the case, i.e., $u\le_{FS} v$, but $u\le_{R} v$ does not hold. Hence, $[u]_{\sim_{R}}\neq [v]_{\sim_{R}}$. Since $Q/_{\sim_{R}}=Q/_{\sim_{FS}}$, we have $[u]_{\sim_{FS}}=[u]_{\sim_{R}}\neq [v]_{\sim_{R}}=[v]_{\sim_{FS}}$.
        Now consider the maximum co-lex orders $\le'$ and $\le''$ for $\mathcal{A}/_{\sim_{R}}$ and $\mathcal{A}/_{\sim_{FS}}$, respectively. Their existence is guaranteed by Lemmas~\ref{4:lm:max_co_lex}~and~\ref{4:lm:FS_max_colex}.
        Moreovewr, note that $Q/_{\sim_{R}}=Q/_{\sim_{FS}}$ implies that the quotient automata $\mathcal{A}/_{\sim_{R}}$ and $\mathcal{A}/_{\sim_{FS}}$ are identical, and thus $\le'$ is equal to $\le''$.
        From $[u]_{\sim_{FS}}\neq [v]_{\sim_{FS}}$ and $u\le_{FS} v$, it follows that $[u]_{\sim_{FS}}\le'' [v]_{\sim_{FS}}$ by definition. On the other hand, it does not hold that $[u]_{\sim_{R}}\le' [v]_{\sim_{R}}$ as $u\le_{R}v$ does not hold, a contradiction.
        
        \item \begin{enumerate}
            \item $(\Rightarrow)$: Since $\le$ is antisymmetric, no distinct $u,v\in Q$ such that $u\le v$ and $v\le u$ exist. Then, the fact that, for all $u,v\in Q$, $u\le v$ holds if and only if $u\le_{R} v$, implies that no distinct $u,v\in Q$ can satisfy both $u\le_R v$ and $v\le_R u$. It follows that $Q=Q/_{\sim_{R}}$ and therefore $|Q|=|Q/_{\sim_{R}}|$.

            $(\Leftarrow)$: Since $|Q|=|Q/_{\sim_{R}}|$, for every distinct $u,v\in Q$ it holds that at most one of $u\le_R v$ and $v\le_R u$ hold. Hence, $\le_R$, the maximum co-lex relation, is antisymmetric and thus it is a maximum co-lex order. Since the maximum co-lex order of an automaton is unique (if it exists), $\le_R$ is equal to $\le$.

            \item $(\Rightarrow)$: By definition $\leq$ satisfies antisymmetry. Thus, if $\leq$ is equal to $\leq_{FS}$, then also $\leq_{FS}$ satisfies antisymmetry. Thus, let us consider the equivalence relation $\sim_{FS}$ induced by $\leq_{FS}$, $u \sim_{FS} v$ can hold if and only if $u = v$, which proves that $Q = Q/_{\sim_{FS}}$, and consequently $\lvert Q \rvert = \lvert Q/_{\sim_{FS}} \rvert$.
            
            $(\Leftarrow)$: We start from the fact that $Q/_{\sim_{R}}$ is a partition of $Q$ thus it holds that $|Q/_{\sim_{R}}|\le |Q|$. Moreover, if $|Q|=|Q/_{\sim_{FS}}|$, then by Theorem~\ref{main:th:1}, we have $|Q|=|Q/_{\sim_{FS}}|\le |Q/_{\sim_{R}}|$. Hence, we obtain $|Q/_{\sim_{R}}|=|Q/_{\sim_{FS}}|=|Q|$. Note that this means that both $Q/_{\sim_{R}}$ and $Q/_{\sim_{FS}}$ contain only singletons; in other words, we obtain $Q/_{\sim_{R}}=Q/_{\sim_{FS}}=Q$. Then the claim follows from (\ref{thm:relationship:r:fs}) and (\ref{thm:relationship:o:r}).\qedhere
        \end{enumerate}
    \end{enumerate}
\end{proof}

We now present a much stronger result: If an NFA $\mathcal{A}$ admits a maximum co-lex order $\le$, then $\le$ must be equal to the maximum co-lex relation $\le_R$ of $\mathcal{A}$. It is worth noting that this relies on the assumption of the unique initial state with no in-coming transitions. In particular, notice that Cotumaccio~\cite{cotumaccio2022graphs} showed that a general labeled graph without this assumption can admit both the maximum co-lex order and the maximum co-lex relation while the two orders are not equal. Here we show that once the unique initial state is assumed, the two orders must be equal once an NFA admits them. To prove this, we introduce the notion of \textit{distance from the source}. Given an NFA $\mathcal{A}=(Q,\delta,\Sigma,s)$, the distance of state $u$ from the source is defined as $\phi(u)=\argmin_{\alpha \in I_u} |\alpha|$.
It is clear that if $\phi(u) = r > 1$ for $u \in Q$, then there must exist $u' \in Q$ and $a \in \Sigma$ such that $u \in \delta_{a}(u')$ and $\phi(u') = r-1$.

\begin{lemma}\label{appb:max_order:lm1}
    Let $\mathcal{A}=(Q,\delta,\Sigma,s)$ be an NFA, let $\leq_{R}$ be the maximum co-lex relation of $\mathcal{A}$, and let $\sim_{R}$ be the equivalence relation induced by $\leq_{R}$. Then, for two distinct states $u, v \in Q$, it holds that $u \sim_{R} v$ implies $\phi(u)=\phi(v)$.
\end{lemma}

\begin{proof}
    Due to Lemma 7 in the article by Cotumaccio~\cite{cotumaccio2022graphs}, we know that $u \sim_{R} v$ holds if and only if, for each preceding pair $(u',v')$ of $(u,v)$, $\lambda(u') = \lambda(v') = \{a\}$, for some $a \in \Sigma$. Let us show the contrapositive of the statement, i.e., suppose that $\phi(u)\neq \phi(v)$ and we will show that this implies the existence of a preceding pair $(u',v')$ of $(u,v)$ with $\lambda(u') \neq \lambda(v')$. Let us consider $\phi(u) = r \neq r' = \phi(v)$, and let us assume, w.l.o.g., $r<r'$, the other case is symmetric. Note that, for each $0 \leq i < r$, we have $u_{r-i} \neq v_{r'-i}$ as $\phi(u_{r-i}) \neq \phi(v_{r'-i})$. The proof proceeds by induction over $r$. If $r=1$, then $(u_{r},v_{r})$ is a preceding pair of $(u,v)$ and $\lambda(s) = \lambda(u_{r}) \neq \lambda(v_{r})$, since $s$ is the only state of $\mathcal{A}$ such that $\lambda(s) = \{\#\}$. Now assume $r > 1$. Then either $\lambda(u_{r}) \neq \lambda(v_{r})$ giving us the required pair of preceding states $\lambda(u_{r}) = \lambda(v_{r}) = \{a\}$ for some $a \in \Sigma$. In this case, by the induction hypothesis, $(u_{r-1},v_{r-1})$ has a preceding pair $(u',v')$ such that $\lambda(u') \neq \lambda(v')$, and, since $u_{r} \in \delta_{a}(u_{r-1})$ and $v_{r'} \in \delta_{a}(v_{r'-1})$, $(u',v')$ is also a preceding pair for $(u,v)$.
\end{proof}

\begin{lemma}\label{appb:max_order:lm2}
    Let $\mathcal{A}=(Q,\delta,\Sigma,s)$ be an NFA, let $\leq_{R}$ be the maximum co-lex relation of $\mathcal{A}$, and let $\sim_{R}$ be the equivalence relation induced by $\leq_{R}$.
    Furthermore, let $u, v \in Q$ be such that $u \sim_{R} v$. Then, for each $u'$ such that $u \in \delta_{a}(u')$ for some $a\in \Sigma$, we have that $\phi(u')=\phi(u)-1$.
\end{lemma}
\begin{proof}
    Suppose for the purpose of contradiction that there exists $u'\in Q$ such that $u \in \delta_{a}(u')$ and $\phi(u') \neq \phi(u)-1$.
    Due to Axiom 1 of Definition~\ref{2:def:col_rel}, it holds that $\lambda(u) = \lambda(v) = \{a\}$, for some $a \in \Sigma$. Then, there exists $v' \in Q$ such that $v \in \delta_{a}(v')$ and $\phi(v')=\phi(v)-1$. Since $u \sim_{R} v$, due to Axiom 2 of Definition \ref{2:def:col_rel}, also $v' \sim_{R} u'$ must hold. However, due to Lemma \ref{appb:max_order:lm1}, we know that $u \sim_{R} v$ implies $\phi(u) = \phi(v)$. Now $\phi(u') \neq \phi(u)-1$ and $\phi(v') = \phi(v)-1$ imply $\phi(u') \neq \phi(v')$. However, using Lemma \ref{appb:max_order:lm1}, $\phi(u') \neq \phi(v')$ implies that $u' \sim_{R} v'$ does not hold, yielding the desired contradiction.
\end{proof}

\begin{proposition}
    Let $\mathcal{A}=(Q,\delta,\Sigma,s)$ be an NFA, and let $\leq_{R}$ be the maximum co-lex relation of $\mathcal{A}$. Then, if $\mathcal{A}$ admits a maximum co-lex order $\leq$, we have that $\leq$ is equal to $\leq_{R}$.
\end{proposition}

\begin{proof}
    Since every co-lex order is also a co-lex relation, and since we know that $\leq_{R}$ satisfies reflexivity and transitivity~\cite[Lemma 5]{cotumaccio2022graphs}, if $\leq_{R}$ satisfies also antisymmetry, then $\leq\,=\, \leq_{R}$ must hold. 
    Therefore, let us suppose for contradiction that $\leq_{R}$ is not antisymmetric. It means that there exist $u,v \in Q$, with $u \neq v$, such that $u \sim_{R} v$, where $\sim_{R}$ is the equivalence relation induced by $\leq_{R}$. By Lemma \ref{appb:max_order:lm1}, we know that if $u \sim_{R} v$ holds, and $u \neq v$, then $\phi(u) = \phi(v)$. Now, let us consider a pair $(\bar{u},\bar{v})$, with $\bar{u},\bar{v} \in Q$ and $\bar{u} \neq \bar{v}$, that satisfies the following properties: (i)~$\bar{u} \sim_{R} \bar{v}$ holds. (ii)~There is no pair $(u^{*},v^{*})$, with $u^{*},v^{*} \in Q$ and $u^{*} \neq v^{*}$, such that $u^{*} \sim_{R} v^{*}$ and $\phi(u^{*}) = \phi(v^{*}) < \phi(\bar{u}) = \phi(\bar{v})$. Due to Axiom 1 of Definition \ref{2:def:col_rel}, we know that $\lambda(\bar{u})=\lambda(\bar{v})=\{a\}$, for some $a \in \Sigma$. Therefore there exist $u',v' \in Q$, such that $\bar{u} \in \delta_{a}(u')$ and $\bar{v} \in \delta_{a}(v')$. Due to Lemma \ref{appb:max_order:lm2}, we know that $\phi(u') = \phi(v') = \phi(\bar{u}) - 1$. Due to Axiom 2 of Definition \ref{2:def:col_rel}, we know that $u' \sim_{R} v'$. We can observe that if $u' \neq v'$, then the pair $(u',v')$ would satisfy the requirements $u' \neq v'$, $u' \sim_{R} v'$, and $\phi(u') = \phi(v') < \phi(\bar{u}) = \phi(\bar{v})$, it follows that $u' = v'$ must hold. Now let us define the two relations $\leq'\; \coloneqq \{(u,u) : u \in Q\}\cup \{(\bar{u},\bar{v})\}$ and $\leq''\; \coloneqq \{(u,u) : u \in Q\}\cup \{(\bar{v},\bar{u})\}$. It follows that both $\leq'$ and $\leq''$ are co-lex orders of $\mathcal{A}$, however $\leq$ cannot be a superset of both relations as it is antisymmetric, a contradiction.\qedhere
\end{proof}

} {

\begin{abstract}
An index on a finite-state automaton is a data structure able to locate specific patterns on the automaton's paths and consequently on the regular language accepted by the automaton itself. Cotumaccio and Prezza [SODA '21], introduced a data structure able to solve pattern matching queries on automata, generalizing the famous FM-index for strings of Ferragina and Manzini [FOCS '00]. The efficiency of their index depends on the width of a particular partial order of the automaton’s states, the smaller the width of the partial order, the faster is the index. However, computing the partial order of minimal width is NP-hard. This problem was mitigated by Cotumaccio [DCC '22], who relaxed the conditions on the partial order, allowing it to be a partial preorder. This relaxation yields the existence of a unique partial preorder of minimal width that can be computed in polynomial time. In the paper at hand, we present a new class of partial preorders and show that they have the following useful properties: (i) they can be computed in polynomial time, (ii) their width is never larger than the width of Cotumaccio's preorders, and (iii) there exist infinite classes of automata on which the width of Cotumaccio's pre-order is linearly larger than the width of our preorder.




\end{abstract}

\keywords{Nondeterministic Finite Automata \and Graph Indexing \and Forward-Stable Partitions \and FM-index.}


\section{Introduction}


The Burrows-Wheeler transform (BWT) is a famous reversible string transformation~\cite{burrows1994block}. While the BWT was initially conceived as a compression tool, it has subsequently been used to implement the FM index~\cite{manziniFM} that is able to locate patterns in almost optimal time, while efficiently compressing the strings at the same time. This indexing strategy has been extended to some finite automata (or directed edge-labeled graphs), the so-called \textit{Wheeler graphs}, by Gagie et al.~\cite{gagie2017wheeler}. Wheeler graphs are a particular class of automata that can be succinctly encoded by a representation that allows pattern matching queries on the strings labeling directed paths in the automaton in almost optimal time. This indexing property relies on the fact that the states of the automaton can be totally ordered in a way that is consistent with the co-lexicographic order of the strings accepted by the automaton's states. Such an order of the states is called a \textit{Wheeler order}. As not all automata admit a Wheeler order, this strategy may fail to apply and hence not all finite automata are Wheeler. Cotumaccio and Prezza~\cite{cotumaccio2021indexing} extended the notion to arbitrary finite automata by not restricting the state order to be total, instead allowing arbitrary partial orders, called \textit{co-lex orders}. As any finite automaton admits such a co-lex order, this enables us to build an efficient index for the language accepted by any automaton in a very similar fashion as one does with the original FM index for a given string. The co-lex order however is not in general unique and, furthermore, the choice of the co-lex order is of crucial importance for the efficiency of the index. More precisely, the efficiency is directly characterized by the \textit{width} of the particular co-lex order: the lower the width of the co-lex order is, the faster pattern matching queries are and the smaller the index is. Unfortunately, computing the co-lex order of minimal width is NP-hard~\cite{cotumaccio2021indexing}. Cotumaccio~\cite{cotumaccio2022graphs} introduced \textit{co-lex relations} by relaxing the requirements on co-lex orders, allowing them to be reflexive relations, i.e., unrequiring antisymmetry and transitivity. This relaxation yields that a unique maximum co-lex relation always exists. Moreover, this maximum is a preorder (i.e., transitivity is recovered through maximality). Maybe most importantly though, the maximum co-lex relation can be computed in polynomial time (more precisely in quadratic time in the number of transitions). The width of the maximum co-lex relation is at most the width of every co-lex order and may be asymptotically smaller. Hence, the index built based on the maximum co-lex relation is never slower and can be asymptotically 
faster than those built based on co-lex orders~\cite{cotumaccio2021indexing}.

\paragraph{Our contribution.}
In this paper, we present a new category of preorders that we call \textit{coarsest forward-stable co-lex (CFS) orders}. CFS orders are equally useful as maximum co-lex relations for constructing indices on automata. Similar to co-lex orders being the generalization of Wheeler orders to width larger than one, CFS orders are the generalization of Wheeler preorders (introduced by Becker et al.~\cite{becker2023sorting}) to width larger than one. As shown by Cotumaccio for maximum co-lex relations, we can show that our CFS orders always exist and are unique. Furthermore, for a given automaton, the unique CFS order can be computed in quadratic time in the number of transitions as well. Most importantly however, we prove that the width of the 
CFS order is never larger than the width of the maximum co-lex relation and, 
in some cases, asymptotically smaller. 

In Figure~\ref{fg:ex:3}, we show the relationships among the previously  mentioned concepts for building indices on automata. All previously known orders shown in Figure~\ref{fg:ex:3} are formally introduced in Section~\ref{sec:not}, where we describe the preliminaries for this work. In Section~\ref{sec:CSF} we introduce coarsest forward-stable co-lex (CFS) orders and in Section~\ref{sec:CSF CR} we prove the claimed properties of these preorders.

\begin{figure}[ht]
\begin{center}
\resizebox{0.9\textwidth}{!}{
\begin{tikzpicture}
[rel/.style={shape = rectangle, align = center, minimum height = 3.75em, minimum width = 7em, font=\footnotesize, draw=black, semithick},
col1/.style={fill=SkyBlue},
col2/.style={fill=YellowOrange},
edges/.style={font=\scriptsize, thick},
widthcol/.style={draw=BlueViolet},
cat/.style={dashed, thin},
catnod/.style={font=\scriptsize}]

    \node[rel, col2] (1) at (9.65,0) {Co-lex relation};
    \node[rel, col2] (2) at (4.9,0) {Co-lex order};
    \node[rel, col2] (3) at (9.65,2.5) {Maximum\\co-lex relation};
    \node[rel, col1] (4) at (4.9,2.5) {Maximum\\co-lex order};
    \node[rel, col2] (5) at (4.9,5) {CFS order};
    \node[rel, col1] (6) at (0.15,2.5) {Wheeler\\preorder};
    \node[rel, col1] (7) at (0.15,0) {Wheeler order};

    \begin{scope}[on background layer]
        \draw[cat] (-1.15, 4.5) rectangle (2.5,-0.75);
        \node[catnod, align = right] at (1.8, 4.15) {\textit{orders of}\\\textit{width 1}};
        \draw[cat] (3.4,-0.75) rectangle (10.95, 6);
        \node[catnod] at (9.2,5.75) {\textit{
        orders of arbitrary width
        }};
    \end{scope}
    
    \draw[->, edges] (1) edge node [below, align = center] {
    is antisymmetric,\\is transitive
    } (2);
    \draw[->, edges] (1) edge node [left, align = right] {
    is the union of\\all co-lex relations
    } (3);
    \draw[->, edges] (2) edge node [left, align = right, fill = white] {
    is the union of\\all co-lex orders
    } (4);
    \draw[->, edges] (3) edge node [above] {
    is antisymmetric
    } (4);
    \path[->, edges] (5) edge node [left, align = right, fill = white] {
    is\\antisymmetric
    } (4);
    
    \draw[->, edges] (3) --  (9.65,5) -- node [below, align = center] {
    induces the coarsest\\forward-stable partition
    } (5);

    \draw[->, edges] (5) -- node [above, align = center]  {
    is a total preorder,\\$\mathcal{A}$ is input consistent
    } (0.15,5) -- (6);
    
    \draw[->, edges] (6) edge node [right, align = left] {
    is a total\\order
    } (7);

    \draw[->, edges] (2) edge node [below, align = center, fill=white]{
    is a total order,\\$\mathcal{A}$ is input\\consistent
    } (7);

    \draw[->, edges, widthcol, bend right = 35] (1) edge (3);

    \draw[->, edges, widthcol, bend right = 35] (2) edge (4);

    \draw[->, edges, widthcol, bend right = 35] (4) edge (5);

    \draw[->, edges, widthcol, bend right = 15] (4) edge (3);

    \draw[->, edges, widthcol] (3) edge (5);
    
\end{tikzpicture}
}
\end{center}
\caption{An NFA $\mathcal{A}$ is \textit{input consistent} if for each state $u$ in $\mathcal{A}$, all incoming edges of $u$ are labeled with the same character. 
We show the connections between the different relations described. 
A relation is orange, if every automaton always admits an instance of that relation, and it is blue otherwise. Orders on the left, i.e., \textit{Wheeler preorders} and \textit{Wheeler orders}, are of width 1, while the others may be of arbitrary width. 
A blue edge $A \rightarrow B$ means that any relation of type $A$ has always a width larger than or equal to a relation of type $B$. A black edge $A \xrightarrow{c} B$ means a relation of type $A$ is also a relation of type $B$ if it satisfies the requirements $c$. In this case, a relation of type $B$ is always also a relation of type $A$ with the following exceptions: (i) The \textit{coarsest forward-stable co-lex order} may not be equal to the \textit{maximum co-lex relation}. (ii) If the \textit{maximum co-lex order} exists, then it may not be equal to the \textit{coarsest forward-stable co-lex order}. (iii) A \textit{Wheeler order} may not be equal to the \textit{Wheeler preorder}. All implications either directly follow from their definitions or are proved in Appendix~\ref{appendix: relationships}.}
\label{fg:ex:3}
\end{figure}

\section{Preliminaries}\label{sec:not}

\paragraph{Nondeterministic finite automata.}
A nondeterministic finite automaton (NFA) is a 4-tuple $(Q,\delta,\Sigma,s)$, where $Q$ represents the set of the states, $\delta : Q \times \Sigma \rightarrow 2^{Q}$ is the automaton's transition function, $\Sigma$ is the alphabet and $s \in Q$ is the initial state. The standard definition of NFAs includes also a set of final states that we omit since we are not concerned in distinguishing between final and non-final states. We assume the alphabet $\Sigma$ to be effective, i.e., each character of $\Sigma$ labels at least one edge of the transition function. A deterministic finite automaton (DFA) is an NFA such that each state has at most one outgoing edge labeled with a given character. Given an NFA $\mathcal{A} = (Q, \delta, \Sigma, s)$, for a state $u \in Q$, and a character $a \in \Sigma$, we use the shortcut $\delta_{a}(u)$ for $\delta(u,a)$. For a set $S \subseteq Q$, we define $\delta_{a}(S) \coloneqq \bigcup_{u \in S}\delta_{a}(u)$. The set of finite strings over $\Sigma$, denoted by $\Sigma^{*}$, is the set of finite sequences of letters from $\Sigma$. We extend the transition function $\delta$ to the elements of $\Sigma^{*}$ in the following way: for $\alpha \in \Sigma^{*}$ and $u \in Q$ we define $\delta(u,\alpha)$ recursively as follows. If $\alpha = \varepsilon$ (i.e. $\alpha$ is the empty string) then $\delta(u,\alpha)  = \{u\}$. Otherwise, if $\alpha = \alpha'a$, with $\alpha' \in \Sigma^{*}$, $a \in \Sigma$, then $\delta(u,\alpha) = \bigcup_{v\in \delta(u,\alpha')}\delta(v,a)$. 


\begin{definition}[Strings reaching a state]\label{2:df:reac_str}
    Given an NFA $\mathcal{A}=(Q,\delta,\Sigma,s)$, the set of strings reaching a state $u \in Q$ is defined as $I_{u} = \{ \alpha \in \Sigma^{*} : u \in \delta(s,\alpha)\}$.
\end{definition}

We say that $I_{u}$ is the regular language recognized by state $u$, while the regular language recognized by $\mathcal{A}$ is the union of all the regular languages recognized by all $\mathcal{A}$'s states. Regarding the NFAs we treat, we make the following assumptions: (i) We assume that every state is reachable from the initial state. (ii) We assume that the initial state has no incoming edges. (iii) We do not require each state to have an outgoing edge for all possible labels. None of these assumptions are restrictive, since any NFA can be modified to satisfy these assumptions without changing its accepted language. Given an NFA $\mathcal{A} = $($Q$, $\delta$, $\Sigma$, $s$), and a state $u \in Q$, $u \neq s$, we denote with $\lambda(u)$ the set of the characters of $\Sigma$ that label the incoming edges of $u$. If $u = s$, we define $\lambda(u) = \{\#\}$, where $\# \notin \Sigma$.

\paragraph{Forward-stable partitions.}

Given a set $U$, a partition $\mathcal{U} = \{U_{i}\}_{i=1}^{k}$ of $U$ is a set of pairwise disjoint non-empty sets $\{U_{1}, ..., U_{k}\}$ whose union is $U$. We call the sets $U_{1}, ... ,U_{k}$ the parts of $\mathcal{U}$. Given two partitions $\mathcal{U}$ and $\mathcal{U}'$, we say that $\mathcal{U}'$ is a refinement of $\mathcal{U}$ if every part of $\mathcal{U}'$ is contained in a part of $\mathcal{U}$. Note that every partition is a refinement of itself. We use the concept of forward-stable partitions, see also the work of Alanko et al. \cite[Section 4.2]{alanko2021wheeler}.

\begin{definition}[Forward-Stability]
Given an NFA $\mathcal{A} = (Q,\delta,\Sigma,s)$ and two sets of states $S$, $T \subseteq Q$, we say that $S$ is \textit{forward-stable} with respect to $T$, if, for all $a \in \Sigma$, $S \subseteq \delta_{a}(T)$ or $S \cap \delta_{a}(T) = \emptyset$ holds. A partition $\mathcal{Q}$ of $\mathcal{A}$'s states is \textit{forward-stable} for $\mathcal{A}$, if, for any two parts $S$, $T \in \mathcal{Q}$, it holds that $S$ is forward-stable with respect to $T$.
\end{definition}

Given an NFA $\mathcal{A} = (Q,\delta,\Sigma,s)$ we say that $\mathcal{Q}$ is the \textit{coarsest} forward-stable partition for $\mathcal{A}$, if for every forward-stable partition $\mathcal{Q}'$ for $\mathcal{A}$, it holds that $\mathcal{Q}'$ is a refinement of $\mathcal{Q}$. It is easy to demonstrate that for each automaton $\mathcal{A}$ there exists a unique coarsest forward-stable partition (a proof can be found in the work of Becker et al. \cite[Appendix A]{becker2023sorting}). An example of the coarsest forward-stable partition for an NFA can be found in Figure \ref{fg:ex:2}. An interesting property about forward-stable partitions is the following: let $\mathcal{Q}$ be a forward-stable partition for an NFA $\mathcal{A}=(Q,\delta,\Sigma,s)$. If $u,v \in Q$ belong to the same part of $\mathcal{Q}$, then $u$ and $v$ are reached by the same set of strings; a proof of this property can be found in \cite[Lemma 4.7]{alanko2021wheeler}. However, the reverse does not necessarily hold: we may have states in different parts of a forward-stable partition $\mathcal{Q}$ that are reached by the same set of strings (even if $\mathcal{Q}$ is the coarsest forward-stable partition of $\mathcal{A}$). An example of this fact can be found in Figure 1 of the article by Becker et al. \cite{becker2023sorting}.  



\paragraph{Relations.}

Given a set $U$, a \textit{relation} $R \subseteq U \times U$ over $U$ is a set of ordered pairs of elements from $U$. For two elements $u$, $v$ from $U$, we write $uRv$ to denote that $(u,v) \in R$. A \textit{partial order} over a set $U$ is a relation that satisfies reflexivity, antisymmetry, and transitivity. If a partial order satisfies also connectedness (a relation over $U$ satisfies connectedness, if for each distinct $u, v \in U$, either $(u,v) \in R$ or $(v,u) \in R$ holds) then it is a \textit{total order}. A \textit{partial  preorder} over a set $U$ is a relation that satisfies reflexivity and transitivity, but not necessarily antisymmetry. A \textit{total preorder} is a partial preorder over a set $U$ that satisfies also connectedness. An \textit{equivalence relation} over a set $U$ is a relation that satisfies reflexivity, symmetry and transitivity. In this paper we use the symbol $\leq$ to denote both partial orders and partial preorders and the symbol $\sim$ for equivalence relations. Given an equivalence relation $\sim$ over a set $U$, we denote with $[u]_{\sim}$ the equivalence class of the element $u \in U$ with respect to $\sim$, i.e., $[u]_{\sim} \coloneqq \{v \in U : u \sim v\}$. We denote with $U/_{\sim}$ the partition of $U$ consisting of all equivalence classes $[u]_{\sim}$, for $u \in U$. We may not specify the set $U$ over which the relation is defined, if it is clear from the context. Given a partial order or a partial preorder $\leq$ over a set $U$, a set $L \subseteq U$ is an \textit{antichain} according to $\leq$ if for every $u, v \in L$, with $u \neq v$, pairs $(u,v)$ and $(v,u)$ do not belong to $\leq$. The \textit{width} of $\leq$, denoted by $\width(\leq)$, is the size of the largest antichain for $\leq$. Note that the width is equal to 1 if and only if $\leq$ is also a total order. A partial preorder $\leq$ over a set $U$ induces an equivalence relation $\sim$ over $U$ in the following way; For $u, v \in U$, we define $u \sim v$ if and only if $u \leq v$ and $v \leq u$. Moreover, a partial preorder $\leq$ over a set $U$ induces a partial order $\leq'$ over the set $U/_{\sim}$ (where $\sim$ is the equivalence induced by $\leq$) defined as $[u]_{\sim} \leq' [v]_{\sim}$ if and only if $u \leq v$. In addition, if $\leq$ is also a total preorder, then $\leq'$ is also a total order, since it is easy to observe that if $\leq$ satisfies connectedness, then also $\leq'$ satisfies connectedness. Given a partial order or a partial preorder $\leq$ over a set $U$, we define the symbol $<$ in the following way: for each $u, v \in U$, $u < v$ holds if and only if $u \leq v$ and $\neg(v \leq u)$. Finally, given an NFA $\mathcal{A} = (Q,\delta,\Sigma,s)$ and an equivalence relation $\sim$ over $Q$, we define the quotient automaton $\mathcal{A}/_{\sim}$ as follows:

\begin{definition}[Quotient automaton]\label{2:def:qa}
    Let $\mathcal{A} = (Q,\delta,\Sigma,s)$ be an NFA and $\sim$ an equivalence relation over $Q$. The \textit{quotient automaton} $\mathcal{A}/_{\sim}=(Q/_{\sim},\delta/_{\sim},\Sigma,s/_{\sim})$ of $\mathcal{A}$ is the NFA defined 
    by letting 
    $Q/_{\sim} := \{[u]_{\sim} : u \in Q\}$, 
    $\delta/_{\sim}([u]_{\sim},a) := \{[v]_{\sim} : (\exists v' \in [v]_{\sim}) (\exists u' \in [u]_{\sim}) (v' \in \delta(u', a))\}$, and 
    $s/_{\sim} := [s]_{\sim}$.
\end{definition}

An example of a quotient automaton of an NFA can be found in Figure \ref{fg:ex:2}.

\begin{figure}[ht]
\begin{center}
\resizebox{0.9\textwidth}{!}{
\begin{tikzpicture}
[dim/.style={minimum size=2.6em}, dots/.style={text centered}, 
scale=1, dim2/.style={minimum size=2.8, scale=0.85}, dots/.style={text centered}]
    \node[initial,state,initial where=below,dim] (0) at (0,-4) {$u_{0}$};
    \node[state,dim] (1) at (-2.3,-3.5) {$u_{1}$};
    \node[state,dim] (2) at (-1,-2.6) {$u_{2}$};
    \node[state,dim] (3) at (1,-2.6) {$u_{3}$};
    \node[state,dim] (4) at (2.3,-3.5) {$u_{4}$};
    \node[state,dim] (5) at (-2.2,-1.5) {$u_{5}$};
    \node[state,dim] (6) at (2.2,-1.5) {$u_{6}$};


    \path [->] (0) edge [bend left = 30] node [below] {$a$} (1);
    \path [->] (0) edge [bend left = 30] node [left] {$a$} (2);
    \path [->] (1) edge [loop left] node [left] {$a$} (1);
    \path [->] (2) edge [loop left] node [left] {$a$} (2);
    \path [->] (0) edge [bend right = 30] node [right] {$a$} (3);
    \path [->] (0) edge [bend right = 30] node [below] {$a$} (4);
    \path [->] (2) edge [bend left = 20] node [above] {$b$} (5);
    \path [->] (2) edge [bend left = 10] node [above] {$b$} (6);
    \path [->] (4) edge [bend right = 10] node [right] {$b$} (6);
    \path [->] (3) edge [bend right = 10] node [above] {$b$} (5);
    \path [->] (5) edge [bend left = 38] node [above] {$b$} (6);
    \path [->] (6) edge [bend right = 19] node [below] {$b$} (5);

    \begin{scope}[xshift=6.3cm]
    \node[initial,state,initial where=below,dim2] (0) at (0,-4) {$[u_{0}]_{\sim}$};
    \node[state,dim2] (1) at (-1.7,-2.5) {$[u_{1}]_{\sim}$};
    \node[state,dim2] (2) at (1.7,-2.5) {$[u_{3}]_{\sim}$};
    \node[state,dim2] (3) at (0,-1.0) {$[u_{5}]_{\sim}$};


    \path [->] (0) edge [bend left = 20] node [below] {$a$} (1);
    \path [->] (0) edge [bend right = 20] node [below] {$a$} (2);
    \path [->] (1) edge [loop left] node [left] {$a$} (1);
    \path [->] (1) edge [bend left = 20] node [above] {$b$} (3);
    \path [->] (2) edge [bend right = 20] node [above] {$b$} (3);
    \path [->] (3) edge [loop below] node [below] {$b$} (3);
    \end{scope}
\end{tikzpicture}
}
\end{center}
\caption{An NFA $\mathcal{A}=(Q,\delta,\Sigma,s)$ on the left. On the right, the corresponding quotient automaton $\mathcal{A}/_{\sim_{FS}}=(Q/_{\sim_{FS}},\delta/_{\sim_{FS}},\Sigma,s_{\sim_{FS}})$ of $\mathcal{A}$ for the coarsest forward-stable partition $Q/_{\sim_{FS}} = \{ \{u_{0}\}, \{u_{1}, u_{2}\}, \{u_{3}, u_{4}\}, \{u_{5}, u_{6}\} \}$.}
\label{fg:ex:2}



%
%

\label{4:fg:ex2}
\end{figure}

\paragraph{Wheeler and Quasi-Wheeler NFAs.}

The notion of Wheeler NFAs was first introduced by Gagie et al. \cite[Definition 1]{gagie2017wheeler}. Wheeler NFAs are a particular class of NFAs that can be endowed with a particular total order of their states that allows them to be efficiently compressed and indexed. From now on, we assume that given an NFA $\mathcal{A}=(Q,\delta,\Sigma,s)$ there exists a total order $\leq$ among the elements of the set $\Sigma \cup \{\#\}$, where $\# < a$ for each $a \in \Sigma$. In this paper we deliberately overload the notation regarding the symbol $\leq$, i.e., we will use the same symbol for orders on integers, the alphabet $\Sigma \cup \{ \# \}$ and states $Q$ of automata. It will still always be clear from the context which order we refer to.

\begin{definition}[Wheeler NFAs]
Let $\mathcal{A} = (Q,\delta,\Sigma,s)$ be an NFA. A Wheeler order $\leq$ of $\mathcal{A}$ is a total order over $Q$ such that state $s$ precedes all states in $Q \setminus \{s\}$, and, for any pair $u \in \delta_{a}(u')$ and $v \in \delta_{a'}(v')$:
\begin{enumerate}
    \item If $a < a'$, then $u < v$.
    \item If $a = a'$ and $u' < v'$, then $u \leq v$. 
\end{enumerate}
We say that $\mathcal{A}$ is a \textit{Wheeler NFA}, if there exists a \textit{Wheeler order} $\leq$ of $\mathcal{A}$

\end{definition}

In other words, A Wheeler order is a total order over $Q$ that sorts the states of $Q$ according to the set strings that reach them (see Definition \ref{2:df:reac_str}). The possibly largest drawback of Wheeler NFAs is that, recognizing whether a given NFA is Wheeler or not is an NP-complete problem~\cite{gibney2019hardness}. For this reason, Becker et al.~\cite[Definition 8]{becker2023sorting} proposed a relaxed version of the problem, introducing the notion of quasi-Wheeler NFAs. At this point, given an NFA $\mathcal{A}=(Q,\delta,\Sigma,s)$, we define as $\sim_{FS}$ the equivalence relation over $Q$, such that for each $u,v \in Q$, $u \sim_{FS} v$ holds if and only if $u$ and $v$ belong to the same part of the coarsest forward-stable partition for $\mathcal{A}$. Consequently, partition $Q/_{\sim_{FS}}$ is the coarsest forward-stable partition for $\mathcal{A}$

\begin{definition}[Quasi-Wheeler NFAs]\label{2:def:qw}
Let $\mathcal{A} = (Q,\delta,\Sigma,s)$ be an NFA. A Wheeler preorder $\leq$ of $\mathcal{A}$ is a total preorder over $Q$ such that:
\begin{enumerate}
    \item $\sim_{FS}$ is the equivalence relation induced by $\leq$.
    \item The quotient automaton $\mathcal{A}/_{\sim_{FS}}$ is a Wheeler NFA, and the total order $\leq'$ induced by $\leq$ is a Wheeler order for $\mathcal{A}/_{\sim_{FS}}$.
\end{enumerate}
We say that $\mathcal{A}$ is a quasi-Wheeler NFA, if there is a Wheeler preorder $\leq$ of $\mathcal{A}$.
\end{definition}

Due to the fact that states in the same part of a forward-stable partition are reached by the same set of strings, it is easy to observe that the automata $\mathcal{A}$ and $\mathcal{A}/{\sim_{FS}}$ of Definition \ref{2:def:qw} recognize the same language. Thus, from an indexing perspective the NFA $\mathcal{A}/_{\sim_{FS}}$ is equally useful as $\mathcal{A}$: an index for $\mathcal{A}/_{\sim_{FS}}$ is also an index for $\mathcal{A}$. Moreover, every Wheeler NFA is also a quasi-Wheeler NFA, and there exist NFAs that are quasi-Wheeler but not Wheeler. This implies that the class of quasi-Wheeler NFAs is strictly larger than the class of Wheeler NFAs (for a formal proof see Lemma 9 and Figure 1 of the article of Becker et al.~\cite{becker2023sorting}). Probably, the most important difference between Wheeler NFAs and quasi-Wheeler NFAs is that the latter can be recognized in polynomial time. In fact, Becker et al. \cite{becker2023sorting} proposed an algorithm that, given in input an NFA $\mathcal{A}$, is able to recognize if $\mathcal{A}$ is quasi-Wheeler in $O(\lvert \delta \rvert \log \lvert Q \rvert)$ time and, if this is the case, compute a Wheeler preorder for $\mathcal A$ in the same running time. This is achieved by extending the partition refinement framework of Paige and Tarjan~\cite{paige1987three} to compute the coarsest forward-stable partition for $\mathcal{A}$ and, at the same time, a total order $\leq$ over the parts of that partition (and thus over $\mathcal{A}$'s states), in such a way that if $\mathcal{A}$ is quasi-Wheeler, then $\leq$ is a Wheeler preorder of $\mathcal{A}$.

\paragraph{$p$-Sortable NFAs.}
Despite the fact that Wheeler NFAs can be indexed and compressed almost optimally, a major limitation regarding Wheeler NFAs is that the class of Wheeler NFAs is very limited. In fact, Wheeler languages, i.e., languages that are recognized by a Wheeler NFA, are star-free and closed only under intersection~\cite{alanko2021wheeler}. Moreover, also the amount of ``nondeterminism'' within a Wheeler NFA is bounded, since every Wheeler NFA admits an equivalent Wheeler DFA  of linear size \cite{alanko2020regular} (while in the general case there is an exponential blow-up of the number of states). For these reasons, Cotumaccio and Prezza \cite[Definition 3.1]{cotumaccio2021indexing} have extended the notion of Wheeler orders by introducing the concept of co-lex orders. From here on, given an NFA $A = (Q,\delta,\Sigma,s)$, and two states $u,v \in Q$, we say that $\lambda(u) \leq \lambda(v)$ holds if and only if for each $a \in \lambda(u)$ and for each $a' \in \lambda(v)$, it holds that $a \leq a'$.

\begin{definition}[Co-lex orders]\label{2:def:col_ord}
    Let $\mathcal{A} = (Q,\delta,\Sigma,s)$ be an NFA. A co-lex order of $\mathcal{A}$ is a partial order $\leq$ over $Q$ that satisfies the following two axioms:
    \begin{enumerate}
        \item For every $u, v \in Q$, if $u < v$, then $\lambda(u) \leq \lambda(v)$.
        \item For every pair $u \in \delta_{a}(u')$ and $v \in \delta_{a}(v')$, if $u < v$, then $u' \leq v'$.
    \end{enumerate}
\end{definition}

A main difference between Wheeler orders and co-lex orders is that every NFA $\mathcal{A}=(Q,\delta,\Sigma,s)$ admits a co-lex order; in fact, the relation $\leq\,\coloneqq \{(v,v) : v \in Q\}$ is a co-lex order of $\mathcal{A}$. Given an NFA $\mathcal{A}$ and a co-lex order $\leq$ of $\mathcal{A}$, we say that $\leq$ is the \textit{maximum} co-lex order of $\mathcal{A}$ if $\leq$ is equal to the union of all co-lex orders of $\mathcal{A}$. In general, an NFA does not always admit a maximum co-lex order. Cotumaccio and Prezza also define the notion of co-lexicographic width of an NFA \cite[Definition 3.3]{cotumaccio2021indexing}:

\begin{definition}[Co-lexicographic width]\label{3:df:co_lex_wdt}
    Let $\mathcal{A}$ be an NFA.
    \begin{enumerate}
        \item The NFA $\mathcal{A}$ is \textit{p-sortable} if there exists a co-lex order $\leq$ of $\mathcal{A}$ of width $p$. 
        \item The \textit{co-lexicographic width} $\bar{p}$ of $\mathcal{A}$ is the smallest $p$ for which $\mathcal{A}$ is p-sortable.
    \end{enumerate}
\end{definition}

From Definition \ref{3:df:co_lex_wdt}, it is possible to observe that if an NFA $\mathcal{A}$ is a Wheeler NFA, then the co-lexicographic width of $\mathcal{A}$ is equal to 1, since a Wheeler order is a particular type of co-lex order that is also a total order. As Cotumaccio and Prezza shows, the co-lexicographic width $\bar{p}$ of an NFA $\mathcal{A}$ measures the space and time complexity with which $\mathcal{A}$ can be encoded and indexed. However, the problem of determining the co-lexicographic width of a given NFA is known to be NP-hard. This follows from NP-hardness of Wheelerness.

\paragraph{Indexable Partial Preorders.}

At this point, we have seen that some particular NFAs $\mathcal{A}=(Q,\delta,\Sigma,s)$ (i.e. Wheeler NFAs) can be endowed with a particular total order over $Q$ (i.e. a Wheeler order) in order to build efficient indexes on top of them. Despite the fact that the problem of determining whether or not a general NFA $\mathcal{A}$ admits a Wheeler order is NP-complete, we have seen that it is possible to compute in polynomial time a particular total preorder $\leq$ over $Q$ such that the quotient automaton $\mathcal{A}/_{\sim_{FS}}=(Q/_{\sim_{FS}},\delta/_{\sim_{FS}},\Sigma,s/_{\sim_{FS}})$ (where $\sim_{FS}$ is the equivalence relation induced by $\leq$) has the following properties: $(i)$ $\mathcal{A}$ and $\mathcal{A}/_{\sim_{FS}}$ accept the same language, $(ii)$ If $\mathcal{A}$ is Wheeler, then also $\mathcal{A}/_{\sim_{FS}}$ is Wheeler, $(iii)$ if $\mathcal{A}/_{\sim_{FS}}$ is Wheeler, then the total order $\leq'$ over $Q/_{\sim_{FS}}$ induced by $\leq$ is a Wheeler order for $\mathcal{A}/_{\sim_{FS}}$. Now it is natural to wonder whether or not these reasonings can be generalized also to arbitrary NFAs; to this end, we introduce the notion of indexable partial preorders.

\begin{definition}[Indexable partial preorders]\label{2:def:indpp}
Let $\mathcal{A} = (Q,\delta,\Sigma,s)$ be an NFA of co-lexicographic width $\bar{p}$, and let $\leq$ be a partial preorder over $Q$. Consider $\mathcal{A}/_{\sim}=(Q/_{\sim},\delta/_{\sim},\Sigma,s/_{\sim})$ the quotient automaton defined by the equivalence relation $\sim$ induced by $\leq$. We say that $\leq$ is an \textit{indexable partial preorder} for $\mathcal{A}$ if the following requirements are satisfied:
\begin{enumerate}
    \item $\mathcal{A}$ and $\mathcal{A}/_{\sim}$ accept the same language.

    \item $\mathcal{A}/_{\sim}$ has co-lexicographic width $\bar{q}$, with $\bar{q} \leq \bar{p}$.

    \item The partial order $\leq'$ over $Q/_{\sim}$ induced by $\leq$ is a co-lex order of $A/_{\sim}$ of width $\bar{q}$.
\end{enumerate}
\end{definition}

It is clear that a Wheeler preorder of an automaton $\mathcal{A}$ is also an indexable partial preorder for $\mathcal{A}$. However, in the literature there exists another example of indexable partial preorders, namely the maximum co-lex relations of Cotumaccio~\cite{cotumaccio2022graphs}. Cotumaccio defined the concept of \textit{co-lex relations} of NFAs~\cite[Definition 1]{cotumaccio2022graphs}. Each NFA $\mathcal{A}$ admits a \textit{maximum} co-lex relation, that is the co-lex relation of $\mathcal{A}$ that is equal to the union of all co-lex relations of $\mathcal{A}$. Finally, Cotumaccio demonstrated that the maximum co-lex relation of an NFA $\mathcal{A}=(Q,\delta,\Sigma,s)$ can be computed in $O(\lvert \delta \rvert^{2})$ time and that it is a partial preorder over $Q$ that satisfies the requirements of Definition \ref{2:def:indpp}. In this work, we propose a new class of indexable partial preorders, which we term \textit{coarsest forward-stable co-lex orders}. Then we demonstrate that the width of our class of indexable partial preorders is always smaller than or equal to the width of the maximum co-lex relation of Cotumaccio, and in some cases, even linearly smaller with respect to the number of the automaton's states.

\section{Coarsest Forward-Stable Co-Lex Orders}\label{sec:CSF}

\paragraph{Co-Lex Relations.}
In this section we define a new category of indexable partial preorders (Definition \ref{2:def:indpp}). We start with defining co-lex relations~\cite[Definition 1]{cotumaccio2022graphs}.

\begin{definition}[Co-lex relations]\label{2:def:col_rel}
    Let $\mathcal{A}=(Q,\delta,\Sigma,s)$ be an NFA. A co-lex relation of $\mathcal{A}$ is a reflexive relation $R$ over $Q$ that satisfies:
    \begin{enumerate}
        \item For every $u, v \in Q$, with $u \neq v$, if $(u,v) \in R$, then $\lambda(u) \leq \lambda(v)$.
        \item For every pair $u \in \delta_{a}(u')$ and $v \in \delta_{a}(v')$, with $u \neq v$, if $(u,v) \in R$, then $(u',v') \in R$.
    \end{enumerate}
\end{definition}

It follows that if a co-lex relation $R$ on an NFA $\mathcal{A}=(Q,\delta,\Sigma,s)$ is also a partial order over $Q$, then $R$ is also a co-lex order of $\mathcal{A}$. Furthermore, every co-lex order of $\mathcal{A}$ is also a co-lex relation of $\mathcal{A}$. Given an NFA $\mathcal{A}$ and a co-lex relation $R$ of $\mathcal{A}$, we say that $R$ is the \textit{maximum} co-lex relation of $\mathcal{A}$ if $R$ is equal to the union of all co-lex relations of $\mathcal{A}$. Although, an NFA $\mathcal{A}$ does not always admit a maximum co-lex order, it has been proved that every NFA admits a maximum co-lex relation, denoted hereafter as $\leq_{R}$, and that $\leq_{R}$ satisfies always transitivity, i.e., $\leq_{R}$ is always a partial preorder~\cite[Lemma 5]{cotumaccio2022graphs}. From here on, given an NFA $\mathcal{A}=(Q,\delta,\Sigma,s)$ we will denote by $\mathcal{A}/_{\sim_{R}}=(Q/_{\sim_{R}},\delta/_{\sim_{R}},\Sigma,s/_{\sim_{R}})$ the quotient automaton of $\mathcal{A}$ defined by the maximum co-lex relation $\leq_{R}$, where $\sim_{R}$ is the equivalence relation over $Q$ induced by $\leq_{R}$. The next lemma shows an important property that characterizes the maximum co-lex relation of an NFA.

\begin{lemma}\cite[Corollary 16]{cotumaccio2022graphs}\label{4:lm:max_co_lex}
Let $\leq_{R}$ be the maximum co-lex relation of an automaton $\mathcal{A}=(Q,\delta,\Sigma,s)$, and let $\mathcal{A}/_{\sim_{R}}=(Q/_{\sim_{R}},\delta/_{\sim_{R}},\Sigma,s/_{\sim_{R}})$ be the quotient automaton of $\mathcal{A}$ defined by $\leq_{R}$. Then the partial order $\leq$ over $Q/_{\sim_{R}}$ induced by $\leq_{R}$ is the maximum co-lex order of $\mathcal{A}/_{\sim_{R}}$.
\end{lemma}

Lemma \ref{4:lm:max_co_lex} is important for our purposes, because it demonstrates that the quotient automaton $\mathcal{A}/_{\sim_{R}}$ always admits a maximum co-lex order. Moreover, in Lemma~\ref{dummy:lemma:1}~(1) we observe that if an NFA $\mathcal{A}$ admits a maximum co-lex order $\leq$, then $\width(\leq)$ is smaller than or equal to the width of any co-lex order of $\mathcal{A}$, since $\leq$ must be a superset of every co-lex order of $\mathcal{A}$. This in turn demonstrates that the co-lexicographic width of $\mathcal{A}$ must be equal to the width of $\leq$. In the next lemma, we demonstrate that the partition $Q/_{\sim_{R}}$ is always a forward-stable partition for $\mathcal{A}$, though we will see later that $Q/_{\sim_{R}}$ may not necessarily be the coarsest forward-stable partition for $\mathcal{A}$.

\begin{restatable}[]{lemma}{lemmaii}\label{4:lm:mx_rl_pr_fw}
    Let $\mathcal{A}=(Q,\delta,\Sigma,s)$ be an NFA and $\leq_{R}$ its maximum co-lex relation. Consider the partition $Q/_{\sim_{R}}$, where $\sim_{R}$ is the equivalence relation over $Q$ induced by $\leq_{R}$. Then $Q/_{\sim_{R}}$ is a forward-stable partition for $\mathcal{A}$.
\end{restatable}

\begin{proof}
    Let $S, T$ be two arbitrary parts of partition $Q/_{\sim_{R}}$. We want to demonstrate that for each $a \in \Sigma$ either $S \subseteq \delta_{a}(T)$ or $S \cap \delta_{a}(T) = \emptyset$ holds. Let $a \in \Sigma$ and suppose for contradiction that there exists $u \in S \cap \delta_{a}(T)$ and $v \in S \setminus \delta_{a}(T)$. Then there exists $u' \in T$ such that $u \in \delta_{a}(u')$ and $S = [u]_{\sim_{R}}$, $T = [u']_{\sim_{R}}$. Since $v$ and $u$ belong to the same part $S$ of $Q/_{\sim_{R}}$, it must hold that $v \leq_{R} u$ and $u \leq_{R} v$. It implies that, due to Axiom 1 of Definition \ref{2:def:col_rel}, $\lambda(v) \leq \lambda(u)$ and $\lambda(u) \leq \lambda(v)$, which can simultaneously hold if and only if $\lambda(v) = \lambda(u) = \{a\}$. It follows that there exists $v' \in Q$ such that $v \in \delta_{a}(v')$. Due to Axiom 2 of Definition \ref{2:def:col_rel}, since $v \leq_{R} u$ we have that $v' \leq_{R} u'$, moreover, since $u \leq_{R} v$ we have that $u' \leq_{R} v'$ and so $u' \sim_{R} v'$. Therefore $v' \in [u']_{\sim_{R}}$ and so $v' \in T$. Therefore, it follows that $v \in \delta_{a}(T)$, a contradiction.
\end{proof}

\paragraph{Quotient Automaton $\mathcal{A}/_{\sim_{FS}}$.}

Given an NFA $\mathcal{A}$, in the following part of this article, we will show some properties of the quotient automaton $\mathcal{A}/_{\sim_{FS}}=(Q/_{\sim_{FS}},\delta/_{\sim_{FS}},\Sigma,s/_{\sim_{FS}})$, where $Q/_{\sim_{FS}}$ is the coarsest forward-stable partition for $\mathcal{A}$. In particular, in the next lemma we show that $\mathcal{A}/_{\sim_{FS}}$ always admits a maximum co-lex order.

\begin{restatable}[]{lemma}{lemmaiii}\label{4:lm:FS_max_colex}
For an NFA $\mathcal{A}=(Q,\delta,\Sigma,s)$ let $\mathcal{A}/_{\sim_{FS}}=(Q/_{\sim_{FS}},\delta/_{\sim_{FS}},\Sigma,s/_{\sim_{FS}})$ be the quotient automaton of $\mathcal{A}$ such that $Q/_{\sim_{FS}}$ is the coarsest forward-stable partition of $\mathcal{A}$ and let $\leq_{R}'$ be the maximum co-lex relation of $\mathcal{A}/_{\sim_{FS}}$. Then $\leq_{R}'$ is the maximum co-lex order of $\mathcal{A}/_{\sim_{FS}}$.   
\end{restatable}

\begin{proof}
    As the set of co-lex relations is a superset of the set of co-lex orders, it suffices to show that $\le_R'$ is a co-lex order. By Definition $\le_R'$ is reflexive and by maximality $\le_R'$ is transitive~\cite[Lemma 5]{cotumaccio2022graphs}. Hence, it remains to show antisymmetry.
    
    Therefore, let us suppose for the purpose of contradiction that $\le_R'$ is not antisymmetric, and let us consider the equivalence relation $\sim_{R}'$ over $Q/_{\sim_{FS}}$ induced by $\leq_{R}'$. Moreover, let us consider the partition $\mathcal{P}$ of $Q$ such that states $u,v \in Q$ belong to the same part of $\mathcal{P}$ if and only if $[u]_{\sim_{FS}} \sim_{R}' [v]_{\sim_{FS}}$ holds. Clearly, partition $Q/_{\sim_{FS}}$ is a refinement of $\mathcal{P}$ since every part of $Q/_{\sim_{FS}}$ is contained in a part of $\mathcal{P}$. However, as $\leq_{R}'$ is not antisymmetric, it is easy to observe that the reverse does not hold, and therefore $\mathcal{P}$ is not a refinement of $Q/_{\sim_{FS}}$. 
    
    At this point, we want to prove that $\mathcal{P}$ is a forward-stable partition for $\mathcal{A}$. Let $S, T$ be two arbitrary parts of $\mathcal{P}$. We have to prove that $\forall a \in \Sigma$, $S \subseteq \delta_{a}(T)$ or $S \cap \delta_{a}(T) = \emptyset$ hold. Let $a \in \Sigma$ and suppose for contradiction that there exists $u \in \delta_{a}(T) \cap S$, and $v \in S \setminus \delta_{a}(T)$, so we know that there exists $u' \in T$, such that $u \in \delta_{a}(u')$. Clearly, $v \notin [u]_{\sim_{FS}}$, since if $v \in [u]_{\sim_{FS}}$, then, due to the fact that $Q/_{\sim_{FS}}$ is forward-stable, it must exist a state $v' \in [u']_{\sim_{FS}}$, such that $v \in \delta_{a}(v')$. Therefore we have $[v]_{\sim_{FS}} \neq [u]_{\sim_{FS}}$. Since $[v]_{\sim_{FS}} \sim_{R}' [u]_{\sim_{FS}}$, and due to Axiom 1 of Definition \ref{2:def:col_rel}, we have that $\lambda([v]_{\sim_{FS}})=\lambda([u]_{\sim_{FS}})=\{a\}$. It follows that there exists a state $[v']_{\sim_{FS}} \in Q/_{\sim_{FS}}$ such that $[v]_{\sim_{FS}} \in \delta/_{\sim_{FS}}([v']_{\sim_{FS}}, a)$. Due to Axiom 2 of Definition \ref{2:def:col_rel}, since $[u]_{\sim_{FS}} \leq_{R}' [v]_{\sim_{FS}}$, it follows that $[u']_{\sim_{FS}} \leq_{R}' [v']_{\sim_{FS}}$, and since $[v]_{\sim_{FS}} \leq_{R}' [u]_{\sim_{FS}}$, it follows that $[v']_{\sim_{FS}} \leq_{R}' [u']_{\sim_{FS}}$; therefore $[v']_{\sim_{FS}} \sim_{R}' [u']_{\sim_{FS}}$ holds.
    Due to the fact that $Q/_{\sim_{FS}}$ is forward-stable, there exists $v'' \in [v']_{\sim_{FS}}$ such that $v \in \delta_{a}(v'')$, moreover, since $[u']_{\sim_{FS}} \sim_{R}' [v'']_{\sim_{FS}}$, it follows that $v'' \in T$, which in turn proves that $v \in \delta_{a'}(T)$. Therefore $\mathcal{P}$ is a forward-stable partition for $\mathcal{A}$; however, $\mathcal{P}$ is not a refinement of $Q/_{\sim_{FS}}$, this contradicts the hypothesis that $Q/_{\sim_{FS}}$ is the coarsest forward-stable partition of $\mathcal{A}$. It follows that it was absurd to assume that $\leq_{R}'$ was not antisymmetric.
\end{proof}

\paragraph{CFS orders.} 
Lemma~\ref{4:lm:FS_max_colex} not only proves that $\mathcal{A}/_{\sim_{FS}}$ always admits a maximum co-lex order, but also that the maximum co-lex order of $\mathcal{A}/_{\sim_{FS}}$ is equal to the maximum co-lex relation of $\mathcal{A}/_{\sim_{FS}}$. 
At this point we have all the ingredients to introduce \textit{coarsest forward-stable co-lex (CFS) orders}.

\begin{definition}[Coarsest forward-stable co-lex (CFS) order]\label{4:def:FS_colex}
    Let $\mathcal{A}=(Q,\delta,\Sigma,s)$ be an NFA, and let $\mathcal{A}/_{\sim_{FS}}=(Q/_{\sim_{FS}},\delta/_{\sim_{FS}},\Sigma,s/_{\sim_{FS}})$ be the quotient automaton of $\mathcal{A}$ such that $Q/_{\sim_{FS}}$ is the coarsest forward-stable partition for $\mathcal{A}$. Let $\leq$ be the maximum co-lex order of $\mathcal{A}/_{\sim_{FS}}$, then we say that the \textit{coarsest forward-stable co-lex (CFS) order} of $\mathcal{A}$, denoted as $\leq_{FS}$, is the (unique) relation over $Q$ such that, for each $u,v \in Q$, $u \leq_{FS} v$ holds if and only if $[u]_{\sim_{FS}} \leq [v]_{\sim_{FS}}$.
\end{definition}

We will now demonstrate some properties of the CFS order of an NFA. First of all, we want to demonstrate that the CFS order is a partial preorder over the automaton's states.

\begin{restatable}[]{lemma}{lemmaiv}\label{4:lm:fs_pp}
    The CFS order $\leq_{FS}$ of any NFA is a partial preorder over its states.
\end{restatable}

\begin{proof}
    We have to prove that $\leq_{FS}$ satisfies reflexivity and transitivity. Let us consider the maximum co-lex order $\leq$ of the quotient automaton $\mathcal{A}/_{\sim_{FS}}=(Q/_{\sim_{FS}},\delta/_{\sim_{FS}},\Sigma,s/_{\sim_{FS}})$. By definition of co-lex order, $\leq$ is a partial order, therefore it satisfies reflexivity and transitivity.
    Let us consider an arbitrary state $u \in Q$. By Definition \ref{4:def:FS_colex}, $u \leq_{FS} u$ if and only if $[u]_{\sim_{FS}} \leq [u]_{\sim_{FS}}$; however $[u]_{\sim_{FS}} \leq [u]_{\sim_{FS}}$ holds since $\leq$ is reflexive, therefore also $u \leq_{FS} u$ must hold. It follows that $\leq_{FS}$ satisfies reflexivity.
    Let us consider three arbitrary states $u,v,z \in Q$ such that $u \leq_{FS} v$ and $v \leq_{FS} z$. We have to prove that $u \leq_{FS} z$. Since $u \leq_{FS} v$ we have that $[u]_{\sim_{FS}} \leq [v]_{\sim_{FS}}$, and since $v \leq_{FS} z$ we have that $[v]_{\sim_{FS}} \leq [z]_{\sim_{FS}}$. Moreover, due to the fact that $\leq$ is transitive, we have also that $[u]_{\sim_{FS}} \leq [z]_{\sim_{FS}}$, and consequently $u \leq_{FS} z$. It follows that $\leq_{FS}$ satisfies transitivity.
\end{proof}

The following remarks follow from $\leq_{FS}$ being a partial preorder.

\begin{restatable}[]{remark}{remarki}\label{4:ob:db_obs}
    For an NFA $\mathcal{A}=(Q,\delta,\Sigma,s)$ let $\mathcal{A}/_{\sim_{FS}}=(Q/_{\sim_{FS}},\delta/_{\sim_{FS}},\Sigma,s/_{\sim_{FS}})$ be the quotient automaton with $Q/_{\sim_{FS}}$ being the coarsest forward-stable partition for $\mathcal{A}$. Moreover, let $\leq_{FS}$ and $\leq$ be the CFS order of $\mathcal{A}$ and the maximum co-lex order of $\mathcal{A}/_{\sim_{FS}}$, respectively. Then (1)~$\sim_{FS}$ is equal to the equivalence relation induced by $\leq_{FS}$, and (2)~$\leq$ is the partial order induced by $\leq_{FS}$.
\end{restatable}

\begin{proof}
    For~(1), observe that $\leq$ is a partial order and thus satisfies antisymmetry. Therefore, given two states $u,v \in Q$, $u \leq_{FS} v$ and $v \leq_{FS} u$ can both hold if and only if $u$ and $v$ belong to the same part of the coarsest forward-stable partition $Q/_{\sim_{FS}}$ for $\mathcal{A}$, i.e., if and only if $u \sim_{FS} v$.
    For~(2), note that $\leq_{FS}$ is a partial preorder over $Q$, while $\leq$ is a partial order over $Q/_{\sim_{FS}}$. Moreover, $u \leq_{FS} v$ holds if and only if $[u]_{\sim_{FS}} \leq [v]_{\sim_{FS}}$ holds.
\end{proof}

\section{Relation between CFS Order and Max Co-Lex Relation}
\label{sec:CSF CR}

At this point, we have introduced the maximum co-lex relation $\leq_{R}$ and the CFS order $\leq_{FS}$ of an NFA. In this section we study the relation between $\leq_{R}$ and $\leq_{FS}$. In particular, we can prove that $\leq_{FS}$ is always a superset of $\leq_{R}$, while the opposite does not necessarily hold. In order to do that, we need first to introduce the notion of \textit{preceding pairs} from the article of Cotumaccio~\cite{cotumaccio2022graphs}.

\begin{definition}[Preceding pairs]\cite[Definition 6]{cotumaccio2022graphs}.\label{4:def:prec_pair}
    Let $\mathcal{A}=(Q,\delta,\Sigma,s)$ be an NFA and let $(u',v'),(u,v) \in Q \times Q$ be pairs of distinct states. The pair $(u',v')$ precedes $(u,v)$ in $\mathcal{A}$ if there are $u_{1},...,u_{r}, v_{1},...,v_{r} \in Q$ with $r \geq 1$ and $a_{1}, ..., a_{r-1} \in \Sigma$ such that:
    \begin{enumerate}
        \item $u_{1} = u'$ and $v_{1} = v'$,
        \item $u_{r} = u$ and $v_{r} = v$,
        \item for all $i = 1, \ldots , r-1$, $u_{i} \neq v_{i}$, and 
        \item for all $i = 1, \ldots , r-1$, $u_{i+1} \in \delta_{a_{i}}(u_{i})$ and $v_{i+1} \in \delta_{a_{i}}(v_{i})$.
    \end{enumerate}
\end{definition}

Note that if $u,v \in Q$ are distinct states, then pair $(u,v)$ trivially precedes itself. Next lemma shows an important relation between the maximum co-lex relation and the preceding pairs of an NFA.

\begin{lemma}\cite[Lemma 7]{cotumaccio2022graphs}\label{4:lm:pr_pairs}
Let $\mathcal{A}=(Q,\delta,\Sigma,s)$ be an NFA, let $\leq_{R}$ be the maximum co-lex relation of $\mathcal{A}$, and let $u,v \in Q$ be distinct states. Then, $u \leq_{R} v$ holds if and only if for all pairs $(u',v')$ preceding $(u,v)$ in $\mathcal{A}$ it holds that $\lambda(u') \leq \lambda(v')$.
\end{lemma}

Before proving that $\leq_{FS}$ is a superset of $\leq_{R}$, we introduce another lemma regarding preceding pairs in quotient automata.

\begin{restatable}[]{lemma}{lemmavi}\label{4:lemma:prec_pair_qa}
    Let $\mathcal{A}=(Q,\delta,\Sigma,s)$ be an NFA, and let $\mathcal{A}/_{\sim}=(Q/_{\sim},\delta/_{\sim},\Sigma,s/_{\sim})$ be a quotient automaton of $\mathcal{A}$ such that $Q/_{\sim}$ is a forward-stable partition for $\mathcal{A}$. Consider two pairs $([u]_{\sim},[v]_{\sim}), ([u']_{\sim},[v']_{\sim})$, with $[u]_{\sim},[v]_{\sim},[u']_{\sim},[v']_{\sim} \in Q/_{\sim}$, such that $([u']_{\sim},[v']_{\sim})$ precedes $([u]_{\sim},[v]_{\sim})$ in $\mathcal{A}/_{\sim}$. Then, there exist $u'',v'' \in Q$, with $u'' \in [u']_{\sim}, v'' \in [v']_{\sim}$, such that $(u'',v'')$ precedes $(u,v)$ in $\mathcal{A}$.
\end{restatable}
\begin{proof}
        Since $([u']_{\sim},[v']_{\sim})$ precedes $([u]_{\sim},[v]_{\sim})$ in $\mathcal{A}/_{\sim}$, we know that there exists a sequence $[u_{1}]_{\sim},...,[u_{r}]_{\sim},[v_{1}]_{\sim},...,[v_{r}]_{\sim}$ that satisfies the requirements of Definition \ref{4:def:prec_pair}. The proof proceeds by induction over the parameter $r$. Firstly, let us prove the statement for $r=1$. We know that $u,v$ are distinct states, because $[u]_{\sim}, [v]_{\sim}$ are distinct, and therefore pair $(u,v)$ trivially precedes itself in $\mathcal{A}$. Now let us prove the statement for a general $r$. Let us consider the sequence $[u_{2}]_{\sim},...,[u_{r}]_{\sim},[v_{2}]_{\sim},...,[v_{r}]_{\sim}$, due to the inductive hypothesis, we know that there exist $u^{*} \in [u_{2}]_{\sim}$ and $v^{*} \in [v_{2}]_{\sim}$, such that $(u^{*},v^{*})$ precedes $(u,v)$ in $\mathcal{A}$. Since $[u_{2}]_{\sim} \in \delta/_{\sim}([u_{1}]_{\sim}, a_{1})$ and $[v_{2}]_{\sim} \in \delta/_{\sim}([v_{1}], a_{1})$, and due to the definition of quotient automaton, we know that there exist $\bar{u}_{1},\bar{u}_{2},\bar{v}_{1},\bar{v}_{2} \in Q$ such that: $\bar{u}_{1} \in [u_{1}]_{\sim}$, $\bar{u}_{2} \in [u_{2}]_{\sim}$, $\bar{v}_{1} \in [v_{1}]_{\sim}$, $\bar{v}_{2} \in [v_{2}]_{\sim}$, $\bar{u_{2}} \in \delta(\bar{u_{1}}, a_{1})$, $\bar{v_{2}} \in \delta(\bar{v_{1}}, a_{1})$. Moreover, since $Q/_{\sim}$ is forward-stable, there must exist also $u'' \in [u_{1}]_{\sim}$ and $v'' \in [v_{1}]_{\sim}$, such that $u^{*} \in \delta(u'',a_{1})$ and $v^{*} \in \delta(v'',a_{1})$ (note that $u'' \neq v''$ because $[u_{1}]_{\sim} \neq [v_{1}]_{\sim}$). Moreover, since $(u^{*},v^{*})$ precedes $(u,v)$ in $\mathcal{A}$, then also $(u'',v'')$ must precede $(u,v)$ in $\mathcal{A}$. It follows that the lemma holds.
\end{proof}

At this point we are ready to prove that $\leq_{FS}$ is a superset of $\leq_{R}$.

\begin{restatable}[]{lemma}{lemmavii}\label{4:lm:fs_rf_cr}
    For an NFA $\mathcal{A} = (Q,\delta,\Sigma,s)$ let $\leq_{R}$ be the maximum co-lex relation of $\mathcal{A}$ and let $\leq_{FS}$ be the CFS order of $\mathcal{A}$. Then $\leq_{FS}$ is a superset of $\leq_{R}$.
\end{restatable}

\begin{proof}
    We suppose for contradiction that $\leq_{FS}$ is not a superset of $\leq_{R}$. Then, there is $(u,v) \in \; \leq_{R}$ with $(u,v) \notin \; \leq_{FS}$. Let $\mathcal{A}/_{\sim_{FS}}=(Q/_{\sim_{FS}},\delta/_{\sim_{FS}},\Sigma,s/_{\sim_{FS}})$ be the quotient automaton of $\mathcal{A}$ such that $Q/_{\sim_{FS}}$ is the coarsest forward-stable partition for $\mathcal{A}$, and let $\leq$ be the maximum co-lex order of $\mathcal{A}/_{\sim_{FS}}$ (which exists due to Lemma \ref{4:lm:FS_max_colex}). By Definition \ref{4:def:FS_colex}, $(u,v) \notin\; \leq_{FS}$ if and only if $([u]_{\sim_{FS}}, [v]_{\sim_{FS}}) \notin\; \leq$. In addition, $[u]_{\sim_{FS}} \neq [v]_{\sim_{FS}}$, since $\leq$ satisfies reflexivity. Due to Lemma 7 in the article by Cotumaccio~\cite{cotumaccio2022graphs}
    and because $\leq$ is also the maximum co-lex relation of $\mathcal{A}/_{\sim_{FS}}$ (see Lemma \ref{4:lm:FS_max_colex}), there exist $[u']_{\sim_{FS}}, [v']_{\sim_{FS}} \in Q/_{\sim_{FS}}$ such that the pair $([u']_{\sim_{FS}}, [v']_{\sim_{FS}})$ precedes $([u]_{\sim_{FS}}, [v]_{\sim_{FS}})$ in $\mathcal{A}/_{\sim_{FS}}$, and $\lambda([u']_{\sim_{FS}}) \leq \lambda([v']_{\sim_{FS}})$ does not hold. Due to Lemma \ref{4:lemma:prec_pair_qa}, there exist $u'' \in [u']_{\sim_{FS}}$, $v'' \in [v']_{\sim_{FS}}$ such that $(u'',v'')$ precedes $(u,v)$ in $\mathcal{A}$. Moreover, since $u'' \in [u']_{\sim_{FS}}$, we have that $\lambda(u'') = \lambda([u']_{\sim_{FS}})$, and analogously  $\lambda(v'') = \lambda([v']_{\sim_{FS}})$. Hence, if $\lambda([u']_{\sim_{FS}}) \leq \lambda([v']_{\sim_{FS}})$ does not hold, then also $\lambda(u'') \leq \lambda(v'')$ does not hold. However, again using Lemma 7 in the article by Cotumaccio~\cite{cotumaccio2022graphs}, this contradicts the assumption $(u,v) \in\;\leq_{R}$. Hence, $\leq_{FS}$ is a superset of $\leq_{R}$.
\end{proof}

The following lemma then allows us to conclude that the width of the CFS order is always at most the width of the maximum co-lex relation.

\begin{restatable}[]{lemma}{dummylemmai}\label{dummy:lemma:1}
    The following two statements hold:
    (1) Let $\leq$ and $\leq'$ be two partial orders or partial preorders over the same set $U$. If $\leq$ is a superset of $\leq'$, then $\width(\leq) \leq \width(\leq')$. (2) Let $\leq$ be a partial preorder over a set $U$ and let $\leq'$ be the partial order induced by $\leq$ over $U/_{\sim}$. Then, $\width(\leq) = \width(\leq')$.
\end{restatable}

\begin{proof}
    (1) Since $\leq$ is a superset of $\leq'$, we know that for each $u,v \in U$, if $(u,v) \notin\, \leq$, then $(u,v) \notin\,\leq'$. It follows that every antichain according to $\leq$ is also an antichain according to $\leq'$. (2) Let us denote with $p$ and $p'$ the integers $\width(\leq)$ and $\width(\leq')$, respectively. Let us consider $L$, an arbitrary antichain according to $\leq$, and let us define the set $L' \coloneqq \{ [u]_{\sim} : u \in L \}$; clearly $\lvert L \rvert = \lvert L' \rvert$. It is easy to observe that $L'$ is an antichain according to $\leq'$, this proves that $p \leq p'$. Now, let us consider $L'$, an arbitrary antichain according to $\leq'$, and let us consider a possible set $L$, in such a way that $L$ contains a representative of each equivalence class $[u]_{\sim}$, such that $[u]_{\sim} \in L'$. Clearly, $L$ is an antichain according to $\leq$, and $\lvert L' \rvert = \lvert L \rvert$. This proves that $p' \leq p$. It follows that $p = p'$.
\end{proof}

We are now ready to prove that CFS orders are a class of indexable partial preorders, see Definition \ref{2:def:indpp}. Moreover, we will show that the quotient automaton $\mathcal{A}/_{\sim_{FS}}$ has never more states than $\mathcal{A}/_{\sim_{R}}$, and that the co-lexicographic width of the former is never larger then the co-lexicographic width of the latter.

\begin{restatable}[]{theorem}{mainthi}\label{main:th:1}
    Let $\mathcal{A}=(Q,\delta,\Sigma,s)$ be an NFA, and let $\leq_{R}$ and $\leq_{FS}$ be the maximum co-lex relation and the CFS order of $\mathcal{A}$, respectively. Finally, let $\mathcal{A}/_{\sim_{R}}=(Q/_{\sim_{R}},\delta/_{\sim_{R}},\Sigma,s/_{\sim_{R}})$ and $\mathcal{A}/_{\sim_{FS}}=(Q/_{\sim_{FS}},\delta/_{\sim_{FS}},\Sigma,s/_{\sim_{FS}})$ be the quotient automata of $\mathcal{A}$ defined by $\leq_{R}$ and $\leq_{FS}$, respectively. Then:
    \begin{enumerate}
        \item $\leq_{FS}$ is an indexable partial preorder for $\mathcal{A}$.
        \item The co-lexicographic width of $\mathcal{A}/_{\sim_{FS}}$ is smaller than or equal to the co-lexicographic width of $\mathcal{A}/_{\sim_{R}}$.
        \item The number of parts in partition $Q/_{\sim_{FS}}$ is smaller than or equal to the number of parts in partition $Q/_{\sim_{R}}$. 
    \end{enumerate}
\end{restatable}
\begin{proof}
\begin{enumerate}
    \item By Lemma \ref{4:lm:fs_pp}, $\leq_{FS}$ is a partial preorder over $Q$, thus it remains to prove the three properties in Definition~\ref{2:def:indpp}. 
    (a)~Clearly, $\mathcal{A}$ and $\mathcal{A}/_{\sim_{FS}}$ accept the same language, since states in the same part of a forward-stable partition are reached by the same set of strings~\cite[Lemma 4.7]{alanko2021wheeler}.
    (b)~Let $\bar{p}$ be the co-lexicographic width of $\mathcal{A}$. We have to prove that the co-lexicographic width of $\mathcal{A}/_{\sim_{FS}}$ is smaller than or equal to $\bar{p}$. 
    Let $\leq$ be the maximum co-lex order of $\mathcal{A}/_{\sim_{FS}}$. Remark~\ref{4:ob:db_obs} yields that $\leq$ is the partial order induced by $\leq_{FS}$ and hence Lemma~\ref{dummy:lemma:1}~(2) implies that $\width(\le) = \width(\le_{FS})$. Then, Lemma~\ref{4:lm:fs_rf_cr} yields that $\leq_{FS}$ is a superset of $\le_{R}$, the maximum co-lex relation of $\mathcal{A}$, and thus Lemma~\ref{dummy:lemma:1}~(1) implies that $\width(\le_{FS}) \le \width(\le_R)$. Now let $\leq'$ be a co-lex order of $\mathcal{A}$ of width $\bar{p}$. Then, $\leq'$ is also a co-lex relation and thus $\leq_{R}$ is a superset of $\leq'$. Applying Lemma~\ref{dummy:lemma:1}~(1) again yields $\width(\leq_{R})\le \width(\leq')=\bar p$. Thus altogether
    $
        \width(\le) 
        \le \bar p
    $.
    (c)~Let $\bar{q}$ be the co-lexicographic width of $\mathcal{A}/_{\sim_{FS}}$, and let $\leq$ be the maximum co-lex order of $\mathcal{A}/_{\sim_{FS}}$. Remark~\ref{4:ob:db_obs} yields that $\leq$ is the partial order induced by $\leq_{FS}$. Hence, we have to prove that $\width(\leq)=\bar{q}$. It is clear that $\width(\leq)\ge \bar q$. For the other direction, let $\le'$ be a co-lex order of  $\mathcal{A}/_{\sim_{FS}}$ with $\width(\leq')= \bar q$. Since $\leq$ is a superset of $\leq'$, Lemma~\ref{dummy:lemma:1}~(2) yields $\width(\leq)\le\width(\le')=\bar{q}$.

    \item Remark~\ref{4:ob:db_obs} yields that the maximum co-lex order of $\mathcal{A}/_{\sim_{FS}}$ is the partial order induced by $\leq_{FS}$. Hence the co-lexicographic width of $\mathcal{A}/_{\sim_{FS}}$ is equal to $\width(\leq_{FS})$ by Lemma \ref{dummy:lemma:1}~(2). Lemma \ref{4:lm:fs_rf_cr} states that $\le_{FS}$ is a superset of $\le_R$ and thus $\width(\leq_{FS}) \le \width(\leq_R)$ according to Lemma~\ref{dummy:lemma:1}~(1). From Lemma~\ref{4:lm:max_co_lex}, we conclude that the partial order $\le$ induced by $\leq_{R}$ is the maximum co-lex order of $\mathcal A/_{\sim_R}$. Again using Lemma~\ref{dummy:lemma:1}~(2) implies $\width(\leq_{R})=\width(\le)$.
    
    \item By Lemma \ref{4:lm:mx_rl_pr_fw}, $Q/_{\sim_{R}}$ is forward-stable for $\mathcal{A}$, while $Q/_{\sim_{FS}}$ is the coarsest forward-stable partition for $\mathcal{A}$. Thus $Q/_{\sim_{R}}$ is a refinement of $Q/_{\sim_{FS}}$ and the number of parts of $Q/_{\sim_{FS}}$ is at most the number of parts of $Q/_{\sim_{R}}$.\qedhere
\end{enumerate}
\end{proof}

\begin{figure}[ht!]
\begin{center}
\resizebox{0.9\textwidth}{!}{
\begin{tikzpicture}
[dim/.style={minimum size=2.5em}, dots/.style={text centered}, 
scale=0.75, dim2/.style={minimum size=1.4em}, dim3/.style={minimum size=2.2em},
dim4/.style={minimum size=2.8, scale=0.75}]
    \begin{scope}
        \node[dots] (-1) at (0.5, 0.0) {(a)};
        
        \node[initial,state,initial where=below,dim] (1) at (2.3,-4) {$u_{1}$};
        \node[state,dim] (2) at (2.0,-2.0) {$u_{2}$};
        \node[state,dim] (3) at (4.0,-2.0) {$u_{3}$};
        \node[state,dim] (4) at (3.7, -4) {$u_{4}$};
        \node[state,dim] (5) at (1.7,0) {$u_{5}$};
        \node[state,dim] (6) at (4.3,0) {$u_{n}$};
        \node[dots]  (0) at (3.0,0) {...};
    
        \path [->] (1) edge node[left] {$a$} (2);
        \path [->] (1) edge node[above] {$b$} (3);
        \path [->] (2) edge node[left] {$a$} (5);
        \path [->] (2) edge node[right] {$a$} (6);
        \path [->] (3) edge node[left] {$a$} (5);
        \path [->] (3) edge node[right] {$a$} (6);
        \path [->] (3) edge node[right] {$b$} (4);
    \end{scope}

    \begin{scope}[xshift=4cm]
        \node[dots] (-1) at (2.0, 0.0) {(b)};
        \node[dim2] (1) at (3,-5.2) {$u_{1}$};
        \node[dim2] (2) at (3,-3.8) {$u_{2}$};
        \node[dim2] (5) at (2.0,-2.4) {$u_{5}$};
        \node[dim2] (6) at (4.0,-2.4) {$u_{n}$};
        \node[dim2] (3) at (3,-1.0) {$u_{3}$};
        \node[dim2] (4) at (3,0.4) {$u_{4}$};
        \node[dots]  (0) at (3.0,-2.4) {...};
        \node[dots] (-2) at (1.6, 0.0) {};
    
        \path [->] (1) edge node[left] {} (2);
        \path [->] (2) edge node[left] {} (5);
        \path [->] (2) edge node[left] {} (6);
        \path [->] (6) edge node[left] {} (3);
        \path [->] (5) edge node[left] {} (3);
        \path [->] (3) edge node[right] {} (4);
    \end{scope}

    \begin{scope}[xshift=8cm]
        \node[dots] (-1) at (1.5, 0.0) {(c)};
        \node[initial,state,initial where=below,dim4] (1) at (2.3,-4) {$[u_{1}]_{\sim}$};
        \node[state, dim4] (2) at (2.2,-2) {$[u_{2}]_{\sim}$};
        \node[state, dim4] (3) at (3.8,-2) {$[u_{3}]_{\sim}$};
        \node[state, dim4] (4) at (3.7,-4) {$[u_{4}]_{\sim}$};
        \node[state, dim4] (5) at (3,0) {$[u_{n}]_{\sim}$};
        \node[dots] (-2) at (1.2, 0.0) {};
    
        \path [->] (1) edge node[left] {$a$} (2);
        \path [->] (1) edge node[right] {$b$} (3);
        \path [->] (2) edge node[left] {$a$} (5);
        \path [->] (3) edge node[right] {$a$} (5);
        \path [->] (3) edge node[right] {$b$} (4);
    \end{scope}
    \begin{scope}[xshift=12cm]
        \node[dots] (-1) at (1.5, 0.0) {(d)};
        \node[dim3] (1) at (3,-5.2) {$[u_{1}]_{\sim}$};
        \node[dim3] (2) at (3,-3.8) {$[u_{2}]_{\sim}$};
        \node[dim3] (5) at (3,-2.4) {$[u_{n}]_{\sim}$};
        \node[dim3] (3) at (3,-1.0) {$[u_{3}]_{\sim}$};
        \node[dim3] (4) at (3.0,0.4) {$[u_{4}]_{\sim}$};
        \node[dots] (-2) at (1.7, 0.0) {};
    
        \path [->] (1) edge node[left] {} (2);
        \path [->] (2) edge node[left] {} (5);
        \path [->] (5) edge node[left] {} (3);
        \path [->] (3) edge node[left] {} (4);
    \end{scope}
\end{tikzpicture}
}
\end{center}
\caption{(a) An NFA $\mathcal{A}=(Q,\delta,\Sigma,s)$ with $\Sigma = \{a,b\}$,
$Q=\{u_{1},...,u_{n}\}$, with $n > 4$, where $u_{1}=s$ is the initial state, and $\delta$ is s.t.\ $u_{2} \in \delta_{a}(u_{1})$, $u_{3} \in \delta_{b}(u_{1})$, $u_{4}\ \in \delta_{b}(u_{3})$ and for each $4 < i \leq n$, $u_{i} \in \delta_{a}(u_{2})$, $u_{i} \in \delta_{a}(u_{3})$. We observe that $\mathcal{A}/_{\sim_{R}}$ is equal to $\mathcal{A}$ itself. Note that the example is not trivial, as states $u_{2}$ and $u_{3}$ may not be merged without changing the language of $\mathcal{A}$ (due to state $u_{4}$). (b) The Hasse diagram of the maximum co-lex order of $\mathcal{A}/_{\sim_{R}}$, where $\{u_{5},...,u_{n}\}$ forms a largest antichain, consequently the co-lexicographic width of $\mathcal{A}/_{\sim_{R}}$ is equal to $n-4$. (c) The automaton $\mathcal{A}/_{\sim_{FS}}$ consisting of five states, where for each $4 < i \leq n$, $u_{i} \in [u_{n}]_{\sim_{FS}}$. (d) The Hasse diagram of the maximum co-lex order of $\mathcal{A}/_{\sim_{FS}}$. Note that this total order is also a Wheeler order of $\mathcal{A}/_{\sim_{FS}}$, i.e., the co-lexicographic width of $\mathcal{A}/_{\sim_{FS}}$ is equal to 1.}
\label{4:fg:ex1}
\end{figure}

We now prove another important result regarding the automata $\mathcal{A}/_{\sim_{FS}}$ and $\mathcal{A}/_{\sim_{R}}$. Namely, that the number of states and the co-lexicographic width of $\mathcal{A}/_{\sim_{FS}}$ can be in some cases arbitrary smaller with respect to those of $\mathcal{A}/_{\sim_{R}}$.

\begin{restatable}[]{theorem}{mainthii}\label{main:th:2}
    There exists a class of NFAs such that, for each NFA $\mathcal{A}=(Q,\delta,\Sigma,s)$ in this class; the quotient automaton $A/_{\sim_{R}}$, defined by the maximum co-lex relation of $\mathcal{A}$, has a number of states and a co-lexicographic width equal to $O(\lvert Q \rvert)$, while the quotient automaton $A/_{\sim_{FS}}$, defined by the CFS order of $\mathcal{A}$, has a number of states and a co-lexicographic width equal to $O(1)$.
\end{restatable}
\begin{proof}
    We describe the structure of this class of NFAs in Figure \ref{4:fg:ex1}. In particular, in this figure we can observe that the automaton $\mathcal{A}/_{\sim_{FS}}$ has always co-lexicographic width equal to 1 and it is a Wheeler NFA.
\end{proof}

Finally, we show that the CFS order of an automaton $\mathcal{A}=(Q,\delta,\Sigma,s)$  can be computed in $O(\lvert \delta \rvert^{2})$ time.

\begin{corollary}\label{unique:cor}
    Let $\mathcal{A}=(Q,\delta,\Sigma,s)$ be an NFA. We can compute the CFS order $\leq_{FS}$ of $\mathcal{A}$ in $O(\lvert \delta \rvert^{2})$ time.
\end{corollary}

\begin{proof}
    It is possible to use the partition refinement framework of Paige and Tarjian \cite{paige1987three} to compute the coarsest forward-stable partition of $\mathcal{A}$ and consequently the quotient automaton $\mathcal{A}/_{\sim_{FS}}=(Q/_{\sim_{FS}},\delta/_{\sim_{FS}},\Sigma,s_{\sim_{FS}})$ in $O(\lvert \delta \rvert \log \lvert Q \rvert)$ time. After that, we can compute the maximum co-lex relation $\leq_{R}'$ of $\mathcal{A}/_{\sim_{FS}}$ using Cotumaccio's algorithm \cite[Theorem 8]{cotumaccio2022graphs} in $O(\lvert \delta \rvert^{2})$ time, where by Lemma \ref{4:lm:FS_max_colex} we know that $\leq_{R}'$ is also the maximum co-lex order of $\mathcal{A}/_{\sim_{FS}}$. Once we have computed $\leq_{R}'$, we can also reconstruct $\leq_{FS}$.
\end{proof}


\section{Conclusion}\label{sec:conc}

After recapitulating all existing orders for building indices on NFAs using a BWT-based strategy, we have introduced \textit{coarsest forward-stable co-lex (CFS) orders}. Such orders are partial preorders over the states of an NFA. After proving the existence and uniqueness of a CFS order for each NFA, we compared these novel orders with the previously known ones that (i)~can be used for indexing an NFA using a BWT-based strategy and (ii)~are guaranteed to exist, namely co-lex orders~\cite{cotumaccio2021indexing} and maximum co-lex relations~\cite{cotumaccio2022graphs}. Notably, in Theorem~\ref{main:th:1} we proved that the width of our novel CFS order is never larger than the width of the previously mentioned orders. Furthermore, in Theorem~\ref{main:th:2} we showed that there exist NFAs for which the width of the forward-stable co-lex order is $O(1)$, while the width of any co-lex order or the maximum co-lex relation is $O(\lvert Q \rvert)$, where $Q$ is the set of states of the NFA. Finally in Corollary~\ref{unique:cor} we proved that the forward-stable co-lex order of any NFA can be computed in polynomial time, more precisely in quadratic time in the number of transitions of the NFA.

\begin{credits}
\subsubsection{\ackname} We would like to thank Nicola Cotumaccio for insightful discussions on the topic of this paper. Ruben Becker, Sung-Hwan Kim, Nicola Prezza, and Carlo Tosoni are funded by the European Union (ERC, REGINDEX, 101039208). Views and opinions expressed are however those of the author(s) only and do not necessarily reflect those of the European Union or the European Research Council. Neither the European Union nor the granting authority can be held responsible for them. 

\end{credits}

\bibliographystyle{splncs04}

\appendix

\section{Relationships between Orders in Figure~\ref{fg:ex:3}}
\label{appendix: relationships}


In this section we prove some properties regarding the different relations presented in Figure~\ref{fg:ex:3}, namely \textit{Wheeler orders}, \textit{Wheeler preorders}, \textit{(maximum) co-lex orders}, \textit{(maximum) co-lex relations}, and \textit{coarsest forward-stable co-lex orders}.

\begin{proposition}
    Let $\mathcal{A}$ be an NFA. Then, the following statements hold true: 
    \begin{enumerate}
        \item If $\mathcal{A}$ admits a maximum co-lex order $\leq$, then the coarsest forward-stable co-lex order of $\mathcal{A}$ may not be equal to $\leq$.
        \item The maximum co-lex relation of $\mathcal{A}$ may not be equal to the coarsest forward-stable co-lex order of $\mathcal{A}$.
        \item If $\mathcal{A}$ admits a Wheeler preorder $\leq$ and a Wheeler order $\leq'$, then $\leq'$ may be not equal to $\leq$.
    \end{enumerate}
\end{proposition}

\begin{proof}
    \begin{enumerate}
        \item An example is shown in Figure \ref{4:fg:ex1}.
        \item An example is shown in Figure \ref{4:fg:ex1}.
        \item Let us consider the NFA $\mathcal{A}=(Q,\delta,\Sigma,s)$ with the following structure: $Q = \{u_{1}, u_{2}, u_{3}\}$, $\Sigma = \{a\}$, $s = u_{1}$ and $u_{2}, u_{3} \in \delta_{a}(u_{1})$. It is possible to check that $\leq\, \coloneqq \{(u_{1}, u_{1}), (u_{1}, u_{2}), (u_{1}, u_{3}), (u_{2}, u_{2}), (u_{2}, u_{3}), (u_{3}, u_{3}) \}$ is a Wheeler order of $\mathcal{A}$. However, 
        \[
            \leq'\,\coloneqq \{(u_{1}, u_{1}), (u_{1}, u_{2}), (u_{1}, u_{3}), (u_{2}, u_{2}), (u_{2}, u_{3}), (u_{3},u_{2}), (u_{3}, u_{3}) \}
        \]
        is a Wheeler preorder of $\mathcal{A}$, and $\leq'\;\neq\;\leq$. \qedhere
    \end{enumerate}
\end{proof}

The following proposition characterizes the case when the maximum co-lex relation and the coarsest forward-stable co-lex order are identical.

\begin{proposition} \label{thm:relationship}
    Let $\mathcal{A}$ be an NFA. Let $\le_{R}$ and $\le_{FS}$ be the maximum co-lex relation and the coarsest forward-stable co-lex order of $\mathcal{A}$, respectively. Then, the following statements hold true:
    \begin{enumerate}
        \item \label{thm:relationship:r:fs} $\le_R$ is equal to $\le_{FS}$ if and only if $Q/_{\sim_{R}}=Q/_{\sim_{FS}}$.
        \item \label{thm:relationship:o} Assume that $\mathcal{A}$ admits a maximum co-lex order $\leq$, then,
        \begin{enumerate}
            \item \label{thm:relationship:o:r} $\le_{R}$ is equal to $\leq$ if and only if $|Q|=|Q/_{\sim R}|$.
            \item \label{thm:relationship:o:fs }$\le_{FS}$ is equal to $\leq$ if and only if $|Q|=|Q/_{\sim FS}|$. \qedhere
        \end{enumerate}
    \end{enumerate}
\end{proposition}
\begin{proof}
    \begin{enumerate}
        \item $(\Rightarrow)$: Note that if, for all $u,v\in Q$, $u\le_{R}~v$ if and only if $u\le_{FS} v$, then it also holds that, for all $u,v\in Q$, $u\sim_{R} v$ if and only if $u\sim_{FS} v$. Consequently, we obtain $Q/_{\sim_{R}}= Q/_{\sim_{FS}}$.

        $(\Leftarrow)$: Assume $Q/_{\sim_{R}}=Q/_{\sim_{FS}}$. We need to prove that, for every $u,v\in Q$, $u\le_{R} v$ holds if and only if $u\le_{FS}v$ holds. The forward direction immediately follows from Lemma~\ref{4:lm:fs_rf_cr}, thus we shall prove that $u\le_{FS}v$ implies $u\le_{R} v$. Suppose for the purpose of contradiction that this is not the case, i.e., $u\le_{FS} v$, but $u\le_{R} v$ does not hold. Hence, $[u]_{\sim_{R}}\neq [v]_{\sim_{R}}$. Since $Q/_{\sim_{R}}=Q/_{\sim_{FS}}$, we have $[u]_{\sim_{FS}}=[u]_{\sim_{R}}\neq [v]_{\sim_{R}}=[v]_{\sim_{FS}}$.
        Now consider the maximum co-lex orders $\le'$ and $\le''$ for $\mathcal{A}/_{\sim_{R}}$ and $\mathcal{A}/_{\sim_{FS}}$, respectively. Their existence is guaranteed by Lemmas~\ref{4:lm:max_co_lex}~and~\ref{4:lm:FS_max_colex}.
        Moreovewr, note that $Q/_{\sim_{R}}=Q/_{\sim_{FS}}$ implies that the quotient automata $\mathcal{A}/_{\sim_{R}}$ and $\mathcal{A}/_{\sim_{FS}}$ are identical, and thus $\le'$ is equal to $\le''$.
        From $[u]_{\sim_{FS}}\neq [v]_{\sim_{FS}}$ and $u\le_{FS} v$, it follows that $[u]_{\sim_{FS}}\le'' [v]_{\sim_{FS}}$ by definition. On the other hand, it does not hold that $[u]_{\sim_{R}}\le' [v]_{\sim_{R}}$ as $u\le_{R}v$ does not hold, a contradiction.
        
        \item \begin{enumerate}
            \item $(\Rightarrow)$: Since $\le$ is antisymmetric, no distinct $u,v\in Q$ such that $u\le v$ and $v\le u$ exist. Then, the fact that, for all $u,v\in Q$, $u\le v$ holds if and only if $u\le_{R} v$, implies that no distinct $u,v\in Q$ can satisfy both $u\le_R v$ and $v\le_R u$. It follows that $Q=Q/_{\sim_{R}}$ and therefore $|Q|=|Q/_{\sim_{R}}|$.

            $(\Leftarrow)$: Since $|Q|=|Q/_{\sim_{R}}|$, for every distinct $u,v\in Q$ it holds that at most one of $u\le_R v$ and $v\le_R u$ hold. Hence, $\le_R$, the maximum co-lex relation, is antisymmetric and thus it is a maximum co-lex order. Since the maximum co-lex order of an automaton is unique (if it exists), $\le_R$ is equal to $\le$.

            \item $(\Rightarrow)$: By definition $\leq$ satisfies antisymmetry. Thus, if $\leq$ is equal to $\leq_{FS}$, then also $\leq_{FS}$ satisfies antisymmetry. Thus, let us consider the equivalence relation $\sim_{FS}$ induced by $\leq_{FS}$, $u \sim_{FS} v$ can hold if and only if $u = v$, which proves that $Q = Q/_{\sim_{FS}}$, and consequently $\lvert Q \rvert = \lvert Q/_{\sim_{FS}} \rvert$.
            
            $(\Leftarrow)$: We start from the fact that $Q/_{\sim_{R}}$ is a partition of $Q$ thus it holds that $|Q/_{\sim_{R}}|\le |Q|$. Moreover, if $|Q|=|Q/_{\sim_{FS}}|$, then by Theorem~\ref{main:th:1}, we have $|Q|=|Q/_{\sim_{FS}}|\le |Q/_{\sim_{R}}|$. Hence, we obtain $|Q/_{\sim_{R}}|=|Q/_{\sim_{FS}}|=|Q|$. Note that this means that both $Q/_{\sim_{R}}$ and $Q/_{\sim_{FS}}$ contain only singletons; in other words, we obtain $Q/_{\sim_{R}}=Q/_{\sim_{FS}}=Q$. Then the claim follows from (\ref{thm:relationship:r:fs}) and (\ref{thm:relationship:o:r}).\qedhere
        \end{enumerate}
    \end{enumerate}
\end{proof}

We now present a much stronger result: If an NFA $\mathcal{A}$ admits a maximum co-lex order $\le$, then $\le$ must be equal to the maximum co-lex relation $\le_R$ of $\mathcal{A}$. It is worth noting that this relies on the assumption of the unique initial state with no in-coming transitions. In particular, notice that Cotumaccio~\cite{cotumaccio2022graphs} showed that a general labeled graph without this assumption can admit both the maximum co-lex order and the maximum co-lex relation while the two orders are not equal. Here we show that once the unique initial state is assumed, the two orders must be equal once an NFA admits them. To prove this, we introduce the notion of \textit{distance from the source}. Given an NFA $\mathcal{A}=(Q,\delta,\Sigma,s)$, the distance of state $u$ from the source is defined as $\phi(u)=\argmin_{\alpha \in I_u} |\alpha|$.
It is clear that if $\phi(u) = r > 1$ for $u \in Q$, then there must exist $u' \in Q$ and $a \in \Sigma$ such that $u \in \delta_{a}(u')$ and $\phi(u') = r-1$.

\begin{lemma}\label{appb:max_order:lm1}
    Let $\mathcal{A}=(Q,\delta,\Sigma,s)$ be an NFA, let $\leq_{R}$ be the maximum co-lex relation of $\mathcal{A}$, and let $\sim_{R}$ be the equivalence relation induced by $\leq_{R}$. Then, for two distinct states $u, v \in Q$, it holds that $u \sim_{R} v$ implies $\phi(u)=\phi(v)$.
\end{lemma}

\begin{proof}
    Due to Lemma 7 in the article by Cotumaccio~\cite{cotumaccio2022graphs}, we know that $u \sim_{R} v$ holds if and only if, for each preceding pair $(u',v')$ of $(u,v)$, $\lambda(u') = \lambda(v') = \{a\}$, for some $a \in \Sigma$. Let us show the contrapositive of the statement, i.e., suppose that $\phi(u)\neq \phi(v)$ and we will show that this implies the existence of a preceding pair $(u',v')$ of $(u,v)$ with $\lambda(u') \neq \lambda(v')$. Let us consider $\phi(u) = r \neq r' = \phi(v)$, and let us assume, w.l.o.g., $r<r'$, the other case is symmetric. Note that, for each $0 \leq i < r$, we have $u_{r-i} \neq v_{r'-i}$ as $\phi(u_{r-i}) \neq \phi(v_{r'-i})$. The proof proceeds by induction over $r$. If $r=1$, then $(u_{r},v_{r})$ is a preceding pair of $(u,v)$ and $\lambda(s) = \lambda(u_{r}) \neq \lambda(v_{r})$, since $s$ is the only state of $\mathcal{A}$ such that $\lambda(s) = \{\#\}$. Now assume $r > 1$. Then either $\lambda(u_{r}) \neq \lambda(v_{r})$ giving us the required pair of preceding states $\lambda(u_{r}) = \lambda(v_{r}) = \{a\}$ for some $a \in \Sigma$. In this case, by the induction hypothesis, $(u_{r-1},v_{r-1})$ has a preceding pair $(u',v')$ such that $\lambda(u') \neq \lambda(v')$, and, since $u_{r} \in \delta_{a}(u_{r-1})$ and $v_{r'} \in \delta_{a}(v_{r'-1})$, $(u',v')$ is also a preceding pair for $(u,v)$.
\end{proof}

\begin{lemma}\label{appb:max_order:lm2}
    Let $\mathcal{A}=(Q,\delta,\Sigma,s)$ be an NFA, let $\leq_{R}$ be the maximum co-lex relation of $\mathcal{A}$, and let $\sim_{R}$ be the equivalence relation induced by $\leq_{R}$.
    Furthermore, let $u, v \in Q$ be such that $u \sim_{R} v$. Then, for each $u'$ such that $u \in \delta_{a}(u')$ for some $a\in \Sigma$, we have that $\phi(u')=\phi(u)-1$.
\end{lemma}
\begin{proof}
    Suppose for the purpose of contradiction that there exists $u'\in Q$ such that $u \in \delta_{a}(u')$ and $\phi(u') \neq \phi(u)-1$.
    Due to Axiom 1 of Definition~\ref{2:def:col_rel}, it holds that $\lambda(u) = \lambda(v) = \{a\}$, for some $a \in \Sigma$. Then, there exists $v' \in Q$ such that $v \in \delta_{a}(v')$ and $\phi(v')=\phi(v)-1$. Since $u \sim_{R} v$, due to Axiom 2 of Definition \ref{2:def:col_rel}, also $v' \sim_{R} u'$ must hold. However, due to Lemma \ref{appb:max_order:lm1}, we know that $u \sim_{R} v$ implies $\phi(u) = \phi(v)$. Now $\phi(u') \neq \phi(u)-1$ and $\phi(v') = \phi(v)-1$ imply $\phi(u') \neq \phi(v')$. However, using Lemma \ref{appb:max_order:lm1}, $\phi(u') \neq \phi(v')$ implies that $u' \sim_{R} v'$ does not hold, yielding the desired contradiction.
\end{proof}

\begin{proposition}
    Let $\mathcal{A}=(Q,\delta,\Sigma,s)$ be an NFA, and let $\leq_{R}$ be the maximum co-lex relation of $\mathcal{A}$. Then, if $\mathcal{A}$ admits a maximum co-lex order $\leq$, we have that $\leq$ is equal to $\leq_{R}$.
\end{proposition}

\begin{proof}
    Since every co-lex order is also a co-lex relation, and since we know that $\leq_{R}$ satisfies reflexivity and transitivity~\cite[Lemma 5]{cotumaccio2022graphs}, if $\leq_{R}$ satisfies also antisymmetry, then $\leq\,=\, \leq_{R}$ must hold. 
    Therefore, let us suppose for contradiction that $\leq_{R}$ is not antisymmetric. It means that there exist $u,v \in Q$, with $u \neq v$, such that $u \sim_{R} v$, where $\sim_{R}$ is the equivalence relation induced by $\leq_{R}$. By Lemma \ref{appb:max_order:lm1}, we know that if $u \sim_{R} v$ holds, and $u \neq v$, then $\phi(u) = \phi(v)$. Now, let us consider a pair $(\bar{u},\bar{v})$, with $\bar{u},\bar{v} \in Q$ and $\bar{u} \neq \bar{v}$, that satisfies the following properties: (i)~$\bar{u} \sim_{R} \bar{v}$ holds. (ii)~There is no pair $(u^{*},v^{*})$, with $u^{*},v^{*} \in Q$ and $u^{*} \neq v^{*}$, such that $u^{*} \sim_{R} v^{*}$ and $\phi(u^{*}) = \phi(v^{*}) < \phi(\bar{u}) = \phi(\bar{v})$. Due to Axiom 1 of Definition \ref{2:def:col_rel}, we know that $\lambda(\bar{u})=\lambda(\bar{v})=\{a\}$, for some $a \in \Sigma$. Therefore there exist $u',v' \in Q$, such that $\bar{u} \in \delta_{a}(u')$ and $\bar{v} \in \delta_{a}(v')$. Due to Lemma \ref{appb:max_order:lm2}, we know that $\phi(u') = \phi(v') = \phi(\bar{u}) - 1$. Due to Axiom 2 of Definition \ref{2:def:col_rel}, we know that $u' \sim_{R} v'$. We can observe that if $u' \neq v'$, then the pair $(u',v')$ would satisfy the requirements $u' \neq v'$, $u' \sim_{R} v'$, and $\phi(u') = \phi(v') < \phi(\bar{u}) = \phi(\bar{v})$, it follows that $u' = v'$ must hold. Now let us define the two relations $\leq'\; \coloneqq \{(u,u) : u \in Q\}\cup \{(\bar{u},\bar{v})\}$ and $\leq''\; \coloneqq \{(u,u) : u \in Q\}\cup \{(\bar{v},\bar{u})\}$. It follows that both $\leq'$ and $\leq''$ are co-lex orders of $\mathcal{A}$, however $\leq$ cannot be a superset of both relations as it is antisymmetric, a contradiction.\qedhere
\end{proof}

}

\end{document}